\documentclass{article}
\usepackage{amsmath}
\usepackage{amssymb}
\usepackage{amsthm}
\usepackage{amsfonts}
\usepackage[PostScript=dvips]{diagrams}

\usepackage{graphicx} 
\usepackage{booktabs}

\usepackage[english]{babel}

\usepackage[latin1]{inputenc}

\def\bbbc{{\mathbb C}}

\def\bbbz{{\mathbb Z}}
\def\bbbd{{\mathbb D}}
\def\bbbt{{\mathbb T}}
\def\bbbo{{\mathbb O}}
\def\bbbi{{\mathbb I}}

\def\ad{{\mathrm{ad}}}

\def\cG{{\cal G}}

\newcommand\pl{+}
\newcommand\mi{-}

\def\cp1{{\mathbb C\mathbb P}^1}
\def\tr{\mathrm{tr\, }}

\def\aut{{\mbox{Aut\,}}}

\newtheorem{Remark}{Remark}
\newtheorem{Example}{Example}
\newtheorem{Lemma}{Lemma}

\newtheorem{Corollary}{Corollary}

\newtheorem{Theorem}{Theorem}
\newtheorem{Definition}{Definition}

\newcommand{\mf}[1]{\mathfrak{#1}}
\newcommand{\gf}[2]{\mathfrak{#1}_{#2}}
\newcommand{\tf}[3]{\mathfrak{#1}_{#2}^{#3}}
\newcommand{\htf}[3]{\mathfrak{\hat{#1}}_{#2}^{#3}}

\newcommand{\ttf}[3]{\mathfrak{\overline{#1}}_{#2}^{#3}}
\newcommand{\ntf}[3]{\mathfrak{\underline{#1}}_{#2}^{#3}}

\newcommand{\mb}[1]{\mathbb{#1}}

\newcommand{\rep}{\pi}

\newcommand{\bpf}{\begin{proof}}
\newcommand{\epf}{$\qed$\end{proof}}

\newcommand{\eex}{$\hfill\diamondsuit$}

\def\cG{{\cal G}}

\newcommand\lie[2]{\mathfrak{#1l}_{#2}(\bbbc)}
\renewcommand{\sl}{\lie{s}{2}}
\newcommand{\gl}{\lie{g}{2}}

\def\cp1{{\mathbb C\mathbb P}^1}

\def\aut{{\mbox{Aut\,}}}

\newcommand{\Zn}[1]{\bbbz/{#1}}

\begin{document}
\bibliographystyle{alpha}

\title{On the Classification of Automorphic Lie Algebras}

\author{Sara Lombardo\\
Department of Mathematics, Faculty of Sciences\\
Vrije Universiteit, De Boelelaan 1081a, 1081 HV Amsterdam, The Netherlands\\
School of Mathematics, Alan Turing Building \\
University of Manchester, Upper Brook Street, Manchester M13 9EP, UK
\and
Jan A. Sanders \\
Department of Mathematics, Faculty of Sciences\\
Vrije Universiteit, De Boelelaan 1081a, 1081 HV Amsterdam, The Netherlands}

\date{}

\maketitle
\numberwithin{Theorem}{section}
\numberwithin{Lemma}{section}
\numberwithin{Definition}{section}
\numberwithin{Remark}{section}
\numberwithin{Corollary}{section}
\numberwithin{Example}{section}
\begin{abstract}
It is shown that the problem of reduction can be formulated in a uniform way using the theory of invariants. This provides  a powerful tool of analysis and it opens the road to new applications of these algebras, beyond the context of integrable systems. Moreover, it is proven that \(\sl\)--Automorphic Lie Algebras associated to the icosahedral group \(\bbbi\), the octahedral group \(\bbbo\), the tetrahedral group  \(\bbbt\), and the dihedral group \(\bbbd_n\) are isomorphic. The proof is based on techniques from classical invariant theory and makes use of Clebsch-Gordan decomposition and transvectants,
Molien functions and the trace-form. 
This result provides a complete classification of \(\sl\)--Automorphic Lie Algebras associated to finite groups when the group representations are chosen to be the same and it is a crucial step towards the complete classification of Automorphic Lie Algebras.
\end{abstract}

\section{Algebraic reductions and Automorphic Lie Algebras}

Many integrable equations are obtained as \emph{reductions} of larger systems. The fact that this is true for many equations of interest in applications makes of the reduction problem one of the central problems in the theory of integrable systems since its early days. 
A wide class of (algebraic) reductions can be studied in terms of \emph{reduction groups} \cite{Mikhailov81}, that is, reductions can be associated to a discrete symmetry group of the corresponding linear problem (Lax Pair), either given by the physical system or simply forced on the solutions. The simplest example of such a symmetry is the conjugation for self-adjoint operators. 
The requirement that a Lax pair is invariant with respect to a reduction group imposes certain algebraic constraints on the Lax operators and therefore it yields a reduction. As an illustration, consider for instance a fairly general Lax pair 
\[
L=\partial _x-X(x,t,\lambda)\, ,\quad    M=\partial _t-T(x,t,\lambda)\, , 
\]
where
\[
X(x,t,\lambda)=X_{0}(x,t)+X_1(x,t)\lambda+X_{-1}(x,t)\frac{1}{\lambda}\, ,
\]
\[
T(x,t,\lambda)=T_{0}(x,t)+T_{1}(x,t)\lambda +T_{-1}(x,t)\frac{1}{\lambda}+
 T_{2}\lambda^2 +T_{-2}(x,t)\frac{1}{\lambda^2}\, 
\]
are \(n \times n\) matrix functions of \(x\), \(t\) and of the spectral parameter \(\lambda\).
The consistency condition \(\psi_{tx}=\psi_{xt}\) implies that (for all
values of  \(\lambda\))
\[X_t-T_x+[X\,,\,T]=0\]
i.e. a system of $5\,n^2$ nonlinear differential equations amongst the entries of the matrices of the Lax pair. A natural question arises:
how to reduce it in a systematic way?
The system can be reduced imposing symmetry conditions: considering, as an example, reductions associated to the dihedral group $\bbbd_n$ the symmetry constraints read  
\[
\begin{array}{ll}
X(\lambda)=S\, X(\omega\lambda)\, S^{-1}\, ,& X(\lambda)=-X^{T}(1/\lambda)\, ,\quad \omega^{n}=1\,,\\
T(\lambda)=S\, T(\omega\lambda)\, S^{-1}\, ,& T(\lambda)=-T^{T}(1/\lambda)\,
\end{array}
\]
(see \cite{Mikhailov81, lm_jpa04} for details). Solutions of the reduced system have been recently investigated in \cite{BM}.

This purely algebraic reduction technique, first formulated by Mikhailov (see \cite{Mikhailov81}) and later developed in \cite{MSY}, \cite{lm_cmp05}, has been successfully applied both in classical (e.g. \cite{MR2362600}, \cite{MR2299840}, \cite{MR1875202}, \cite{MR1861496}, \cite{MR801573}, \cite{lm_jpa04}) and quantum integrable systems theory (e.g. 
\cite{b1}, \cite{b2}) and it essentially consists in finding invariants  elements of the Lie algebra over the ring of rational functions used in the Lax pair representation of an integrable equation. These 
invariant elements form an infinite dimensional Lie algebra known as \emph{Automorphic Lie Algebra} \cite{lm_cmp05}.
Indeed, a reduction group can be seen also as a group representation $\cG$ of a (sub)group $G$ of the group of automorphisms of the  infinite dimensional Lie algebra underlying the Lax Pair, e.g. $G\subset Aut(\bbbc(\lambda)\otimes\tf{g}{}{})$, where $\tf{g}{}{}\subset\mathfrak{gl}_n$ is a finite dimensional semisimple Lie algebra and $\bbbc(\lambda)$ is the field of rational functions in the complex variable $\lambda$ with values in $\bbbc$ (see \cite{lm_cmp05} for details). In this latter framework reductions corresponding to $\cG$ are  nothing but a restriction of the Lax Pair to the corresponding \emph{automorphic} (i.e. invariant) subalgebras $(\bbbc(\lambda)\otimes\tf{g}{}{})^G$.
In other words, given a reduction group $G$, the
Lie  algebra $(\bbbc(\lambda)\otimes\tf{g}{}{})^G$ is called automorphic, if its elements  $a\in(\bbbc(\lambda)\otimes\tf{g}{}{})^G$ are invariant, $g(a)=a$, with respect to all automorphisms $g\in G$, i.e.$$(\bbbc(\lambda)\otimes\tf{g}{}{})^G=\{ a\in(\bbbc(\lambda)\otimes\tf{g}{}{})\, \,|\,\, \phi (a)=a\, ,\, \forall 
\phi\in G\subset \aut (\bbbc(\lambda)\otimes\tf{g}{}{})\}\,.$$

Originally motivated by the problem of reduction of Lax pairs, Automorphic Lie Algebras are interesting objects in their own right. So much that a classification is now both a mathematical question and a tool in the reduction problem, and therefore in applications to the theory of integrable systems and beyond. 
A first step towards  classification of Automorphic Lie Algebras was presented in \cite{lm_cmp05}, where automorphic algebras associated  to finite groups where considered. These groups are the five groups of Klein's classification, namely, the cyclic group $\Zn{n}$, the dihedral group $\bbbd_n$, the tetrahedral group $\bbbt$, the octahedral group $\bbbo$ and the icosahedral group $\bbbi$. In the paper \cite{lm_cmp05} the authors classify automorphic algebras associated to the dihedral group $\bbbd_n$,  starting from the finite dimensional algebra $\sl$. Examples of automorphic Lie algebras based on $\lie{s}{3}$ were also discussed.

The aim of this paper is to complete the classification for the case \(\tf{g}{}{}=\sl\) and sketch a classification programme for automorphic Lie algebras associated to finite groups more in general. A key feature of this approach is the study of these algebras in the context of classical invariant theory. Indeed, the problem of reduction can be formulated in a uniform way using the theory of invariants. This gives us  a powerful tool of analysis on one hand,  on the other it opens the road to new applications of these algebras, beyond the context of integrable systems. Moreover, it turns out that in the explicit case we present here, where the underlying Lie algebra is
\(\sl\), we can compute Automorphic Lie Algebras only using geometric data.

The paper is organised as follows: in the introduction we recall basic definitions and facts from the theory of Automorphic Lie Algebras and motivate our further investigation; we define the set up in Section \ref{sec:setup}. Section \ref{sec:trans} describes the tool of transvection and recalls basic notions from invariant theory; here we prove a few lemmas in full generality for later use. Section \ref{sec:alias} deals with the classification of \(\sl\)--Automorphic Lie Algebras associated to finite groups. The framework  is defined in Sections \ref{sec:traceform}, \ref{decomposition} and \ref{sec:groups} while Section \ref{sec:structure} describes the general form of the structure constants. In Section \ref{sec:homo} we define homogeneous elements, therefore we go from the variables  \(X\)  and \(Y\)  to  \(\lambda=\frac{Y}{X}\) (or \(\lambda=\frac{X}{Y}\)); fixing an orbit \(\Gamma\) of the group \(G\) we define a basis of homogeneous elements with \(\Gamma\)--divisors and their structure constants (Section \ref{sec:homobasis}). The \emph{normal form} of  \(\sl\)--Automorphic Lie Algebras is given in Section \ref{sec:normal}. Section \ref{sec:averaging} compares the results with the averaging method, while Section \ref{sec:explicit} gives explicit expressiosn for homogeneous bases associated to finite groups. 
The last Section \ref{sec:conclusion} contains concluding remarks while the Appendices contain the case of \(\Zn{n}\) (see Appendix \ref{AppZN}) and the detailed derivation of homogeneous elements (see Appendices \ref{sec:alphadiv}). The case a of different \(\Gamma\)--divisor is the content of the last Appendix \ref{sec:betadivall}.

\section{Set up}\label{sec:setup}
We consider automorphic algebras in the context of classical invariant theory, that is the study of invariants of the action of \(\sl\) on binary forms. It turns out that in the explicit case we present here, where the underlying Lie algebra is \(\sl\), we can compute the Automorphic Lie Algebra only using geometric data. Let us define the set up:

\begin{itemize}
\item Let \(\tf{g}{}{}\)  be a Lie algebra;
\item Let \(\mb{C}[X,Y]\) be the ring of polynomials in \(X\)  and \(Y\) (later we will define \(\lambda=\frac{X}{Y}\) or \(\lambda=\frac{Y}{X}\));
\item Consider \(\htf{g}{}{}=\mb{C}[X,Y]\otimes\tf{g}{}{}\) ;
\item Let \(G\)  be a finite group;
\item Let \(\sigma\)  be a faithful projective representation of \(G\)  in \(\bbbc^2\); 
this restricts \(G\)  to be one of the following types of groups \cite{MR0080930,MR1315530}: 
\begin{equation}\label{eq:list}
\Zn{m} \,,\quad\bbbd_m \,,\quad \bbbt\, ,\quad \bbbo\, ,\quad \bbbi,
\end{equation}
i.e. the cyclic group $\Zn{m}$, the dihedral group $\bbbd_m$, the tetrahedral group $\bbbt$, the octahedral group $\bbbo$ and the icosahedral group $\bbbi$; 
\item Let \(\tau_k\)  be a faithful projective representation of \(G\)  in \(\bbbc^k\);
\item  Let \(\rep\) be a linear representation of \(G\) in \(\tf{g}{}{}\); that is to say that
\(\pi(g)[m_1\,,\,m_2]=[\pi(g) m_1\,,\, \pi(g) m_2]\) \(\forall\, m_i \in \tf{g}{}{}\);
\item \(\sigma\otimes\pi\) defines a projective representation of \(G\) in \(\htf{g}{}{}=\mb{C}[X,Y]\otimes\tf{g}{}{}\).
\end{itemize}

In a diagram the situation can be represented as follows:
\begin{diagram}
G&\rTo{\sigma}&GL(2,\bbbc)\\
\dTo{\pi} \\
Aut(\tf{g}{}{})
\end{diagram}

Let \( g\in G\), \(\, M=p(X,Y)\otimes m\in \htf{g}{}{}\). In general we write
\[ g\,M=\sigma\otimes\pi(g)\, M=p(\sigma(g^{-1})(X,Y))\otimes\tau_k(g)\, m\,\tau_k(g^{-1})\,.\]
In this paper however we will be concerned with the case \(\tf{g}{}{}=\sl\) and we take \(k=2\) and 
\(\tau_2= \sigma\).
This simplifying assumption can no longer be made when considering higher dimensional Lie algebras.
See the discussion in Section \ref{sec:conclusion}.

Let \(\chi\) be a one-dimensional representation.
\begin{Definition}
Let \(V\) be a \(G\)-module. Then \(v\in V\) is \(\chi\)--covariant
if \( g\,v=\chi (g)\, v\), \(\,\forall g\in G\). If \(\chi \) is trivial then \( v\) is called invariant.
\end{Definition}
In the literature, covariants are also called semi-invariants, or relative invariants.
We denote the space of \(\chi\)--covariants as
\(V_G^{\chi}\,\) and the space of covariants by \(V_G\).
Our first goal is to compute the Lie algebra of covariants \(\htf{g}{G}{}\,\).

\subsection{$\tf{g}{}{}=\sl$}
Let \(\tf{g}{}{}\)  be \(\sl\) and let \(\gf{e}{+}\,,\gf{e}{-}\,,\gf{e}{0}\) be a basis, obeying the relations:
\[[\gf{e}{+}\,,\gf{e}{-}]=\gf{e}{0}\,,\quad [\gf{e}{0}\,,\gf{e}{\pm}]=\pm 2 \gf{e}{\pm}\,.\]

%
\begin{Theorem}\label{theo:formA}
For all \(G\) in list (\ref{eq:list}),
\(
\mf{A}=Y^2 \gf{e}{-}+ XY \,\gf{e}{0}-X^2 \gf{e}{+}\in\htf{g}{}{G}\), i.e. is invariant.
\end{Theorem}

\begin{Remark} 
The theorem holds true in the case that \(G=SL(2,\bbbc)\), that is
\[
\tau_2=
\left(
\begin{array}{cc}
 a &   b  \\
 c  &  d
\end{array}
\right)\,,\qquad 
\sigma\left(
\begin{array}{c}
X\\
Y
\end{array}
\right)=
\left(
\begin{array}{c}
aX+bY\\
cX+dY
\end{array}
\right)\,.
\]
\end{Remark}
\begin{proof}
Since \((\ad(\gf{e}{-})-X\frac{\partial}{\partial Y})\mf{A}=0\)
and \((\ad(\gf{e}{+})-Y\frac{\partial}{\partial X})\mf{A}=0\), the form is invariant under the usual
action of \(SL_2(\bbbc)\) on binary forms.
\end{proof}
\begin{Remark} This is just the adjoint representation of \(\sl\);
it is not clear to us how this object could be conceived in \(\lambda\)-language,
since it cannot be associated to a homogeneous invariant element,
as is shown  in section \ref{sec:homo}.
See however Remark \ref{Schwarzian}, Section \ref{sec:trans}.
\end{Remark}

\section{Transvectants}\label{sec:trans} 
In classical invariant theory the basic computational tool is the {\bf transvectant}: given any two covariants, it is possible to construct a number of (possibly) new covariants by computing transvectants. As a simple example consider two linear forms \(aY+bX\), \(cY+dX\); their first transvectant is the determinant of the coefficients, i.e. \(ad-cb\). Similarly, the discriminant \( a_0 a_2-a^{2}_{1}\) of a quadratic form \(a_0 Y^2+2 a_1 X Y+a_2 X^2\) is the second transvectant of the quadratic form itself. 

In this section we will adapt the idea of transvection to compute invariant algebras. We  start from the classical work by Klein about automorphic functions 
and generalise it to the context of automorphic algebras. To do so, we need first to recall some definitions and facts about transvectants and generalise some of the concepts to the present set up.
\begin{Definition}
A {\bf groundform} is a covariant \(\alpha\) with its divisor of zeros equal to an exceptional (or degenerate) orbit.
\end{Definition}
\begin{Remark}
The terminology used here is explained in \cite[II.6]{MR845275}.
\end{Remark}

Let \(\alpha\in\mb{C}_{n}[X,Y]^{G}_{\chi}\) and let \(\alpha_{k,l}=\frac{\partial^{k+l}\alpha}{\partial X^k\partial Y^l}\);
we define the \(k\)th--transvectant of \(\alpha\) with an arbitrary form 
\(\mf{A}\in\htf{g}{G}{}\) as
\[\tf{A}{\alpha}{k}=(\alpha\,,\mf{A} )^k =\sum_{i=0}^k (-1)^i\binom{k}{i}  \alpha_{i,k-i} \tf{A}{k-i,i}{}\,,\qquad\tf{A}{\alpha}{k}\in\htf{g}{G}{\chi}\]

\begin{Example}[Classical Invariant Theory]
In the definition above  \(\mf{A}\) could as well belong to \(\mb{P}_{m}[X,Y]\). It follows from the classical theory \cite{MR0080930,MR1315530} that if  \(G\)  is either \(\bbbt\), \(\bbbo\) or \(\bbbi\) then the groundforms are given by 
\[\alpha\,,\quad \beta=(\alpha,\alpha)^{2}\,,\quad\gamma=(\alpha,\beta)^{1}\,.\]
If one denotes the degree of a form \(\alpha\) by \(\omega_\alpha\) it follows that (see Table \ref{table_covariants})
\[\omega_\beta=2\,\omega_\alpha-4\,,\quad \omega_\gamma=3\,\omega_\alpha-6\,.\]
If \(G\)  is \(\bbbd_m\) then \(\beta\neq(\alpha,\alpha)^2\) and it has to be computed independently (see Table \ref{table_covariants_DN}).
The degree of \(\beta\) is the number of faces  of the Platonic solid and determines its name.
We observe that \(\omega_\alpha-\omega_\gamma+\omega_\beta=2\), the Euler index.
\eex
\end{Example}
\begin{center}
\begin{table}[h!] 
\begin{center}
\begin{tabular}{|c|c|c|c|} \hline
\(G\)  & \(\omega_\alpha\) &  \(\omega_\beta=2\omega_\alpha-4 \) & \(\omega_\gamma=3\omega_\alpha-6\) \\ 
\hline \hline 
\( \bbbi\) & 12 & 20 & 30\\ 
\hline
\( \bbbo\) & 6 & 8 & 12\\ 
\hline 
\( \bbbt\) & 4 & 4 & 6\\ 
\hline
\end{tabular}
\end{center}
\caption{Degrees of the  groundforms of \(\bbbi, \bbbo, \bbbt\)}
\label{table_covariants}
\end{table}
\end{center}
\begin{center}
\begin{table}[ht!]
\begin{center}
\begin{tabular}{|c|c|c|c|} \hline
\(G\)  & \(\omega_\alpha\) &  \(\omega_\beta \) & \(\omega_\gamma=\omega_\alpha+\omega_\beta-2\) \\ 
\hline \hline 
\( \bbbd_m\) & 2 & m & m\\ 
\hline
\end{tabular}
\end{center}
\caption{Degrees of the groundforms of \(\bbbd_m\)}
\label{table_covariants_DN}
\end{table}
\end{center}

\begin{Remark}\label{Schwarzian}
In \(\lambda\)-language, \(\beta\) corresponds to the Schwarzian of \(\alpha\).
\end{Remark}
\begin{Lemma}\label{Lemma:lem1}
\[
\tf{A}{\alpha}{k}=
k!X^{-k}\sum_{l=0}^k 
(-1)^{l} 
\binom{\omega_\alpha-k+l}{l} 
\binom{\omega_\mf{A}-l}{k-l}
\alpha_{0,k-l}
\mf{A}_{0,l}
\]
\end{Lemma}
\begin{proof}
See \cite[page 90]{bluolver}, and references \cite{gundelfinger} and \cite{ovsienko} therein.
\end{proof}

\begin{Lemma} Let \(\mf{A}\) be an invariant quadratic form and let \(f \in \bbbc[X,Y]\) with \(\omega_{f}\geq 2\). Then one has
\begin{eqnarray}
\tf{A}{f}{1}
&=&
X^{-1}\left(
\omega_\mf{A} f_{Y} \mf{A}
-
\omega_f f \mf{A}_{Y}
\right)
\\\tf{A}{f}{2}
&=&2X^{-2}
\left(
\binom{\omega_\mf{A}}{2} f_{YY} \mf{A}
-
(\omega_f-1) (\omega_\mf{A}-1) f_{Y} \gf{A}{Y}
+
\binom{\omega_f}{2} f \gf{A}{YY}
\right),\\
\tf{A}{f}{i} &=&0\mbox{ for } i\geq 3 \mbox{ if $\mf{A}$ is quadratic}.
\end{eqnarray}
\end{Lemma}
\begin{proof}
It follows immediately from Lemma \ref{Lemma:lem1}.
\end{proof}

\begin{Example}\label{D_2ex1}
Let \(\mf{A}\)  be the invariant form given in Theorem \ref{theo:formA}  and let \(f\) be \(\alpha=XY\), which it happens to be a groundform of the dihedral group \(\bbbd_m\) (see Section \ref{sec:DN}). It follows then that
\begin{eqnarray*}
\rho(\tf{A}{\alpha}{1})
&=& \left(
\begin{array}{cc}
0 & \mi 2\,X^2\\
\mi 2\,Y^2 & 0
\end{array}
\right),\\
\rho(\tf{A}{\alpha}{2})
&=& \left(
\begin{array}{cc}
\mi 2 & 0\\
0 & \pl 2
\end{array}
\right),\\
\rho(\tf{A}{\alpha}{i}) &=&0\mbox{ for } i\geq 3\,
\end{eqnarray*}
where \(\rho\) is the standard representation in \(\gl\). \quad\eex
\end{Example}
Let \(\tf{g}{}{}\) be the Lie algebra \(\sl\) with the usual basis and commutation relations:
\begin{eqnarray}
[\gf{e}{0},\gf{e}{\pm}]&=&\pm 2 \gf{e}{\pm}\\
{[}\gf{e}{+},\gf{e}{-}{]}&=& \gf{e}{0}.
\end{eqnarray}
Then the coadjoint representation can be identified with the invariant form \(\mf{A}\) of Theorem \ref{theo:formA}.
\begin{Lemma}
\(X^{-1}[\mf{A},\gf{A}{Y}]=- 2\gf{A}{}\),
\( X^{-1}[\mf{A},\gf{A}{YY}]=- 2\gf{A}{Y}\) and 
\(X^{-1}[\gf{A}{Y},\gf{A}{YY}]=- 2\gf{A}{YY}\).
\end{Lemma}
\begin{proof}
We only prove the first relation, the other proofs are even simpler.
\begin{eqnarray*}
X^{-1}[\mf{A},\gf{A}{Y}]&=&
X^{-1}[ Y^2\gf{e}{-}+ XY \gf{e}{0}-X^2 \gf{e}{+}, 2Y\gf{e}{-}+ X\gf{e}{0}]
\\&=&
2Y[Y \gf{e}{0}-X \gf{e}{+}, \gf{e}{-}]
+ [ Y^2\gf{e}{-}-X^2 \gf{e}{+}, \gf{e}{0}]
\\&=&
Y^2[\gf{e}{0}, \gf{e}{-}]
-2XY[\gf{e}{+}, \gf{e}{-}]
- X^2 [\gf{e}{+}, \gf{e}{0}]
\\&=&
- 2Y^2\gf{e}{-}
-2XY\gf{e}{0}
+ 2X^2 \gf{e}{+}
\\&=&- 2\gf{A}{}
\end{eqnarray*}
\end{proof}
\begin{Corollary}
If we give \(\mf{A}\) and the \(Y\) derivative odd grading, then the \(\tf{A}{}{}, \tf{A}{Y}{}\) and \(\tf{A}{YY}{}\) form a \(\Zn{2}\)-graded Lie algebra,
with grading \(|\tf{A}{}{}|=1\), \(|\tf{A}{Y}{}|=0\), \(|\tf{A}{YY}{}|=1\).
\end{Corollary}
\begin{Theorem}\label{Theorem:THEOGR}
\(\tf{A}{f}{j}\), \(j=0,1,2\) span a \(\Zn{2}\)-graded Lie algebra.
\end{Theorem}
\begin{proof}
Express \(\tf{A}{f}{j}\), \(j=1,2\) in \(\tf{A}{}{}\), \(\tf{A}{Y}{}\) and \(\tf{A}{YY}{}\).
The induced grading is \(|\tf{A}{f}{j}|=1+|f|+j \mod 2\), 
where \(|f|\) is to be defined in Definition \ref{def:grading}.
\end{proof}
\begin{Remark}\label{rem:sl2}
The research leading to this paper started with the observation that if one defined
(using transvection) 
\begin{eqnarray*}
\tf{A}{0}{}&=&-2Y\alpha_X \,\gf{e}{-}-(X\alpha_X-Y\alpha_Y) \,\gf{e}{0}-2X\alpha_Y\, \gf{e}{+},\\
\tf{A}{-}{}&=&Y^2 \gf{e}{-}+ XY \,\gf{e}{0}-X^2\, \gf{e}{+},\\
\tf{A}{+}{}&=&-\alpha_X^2\, \gf{e}{-}+\alpha_X\alpha_Y \,\gf{e}{0}+\alpha_Y^2\, \gf{e}{+},
\end{eqnarray*}
then one has
\begin{eqnarray*}
[\tf{A}{+}{},\tf{A}{-}{}]&=&\omega_\alpha\,\alpha \,\tf{A}{0}{},\\
{[\tf{A}{0}{}, \tf{A}{\pm}{}]}&=&\pm 2 \omega_\alpha\,\alpha\,  \tf{A}{\pm}{}.
\end{eqnarray*}
Using the methods described in Section \ref{sec:homo} one can now construct an Automorphic Lie Algebra.
This observation was presented at the NEEDS 2009 conference.
What it is not clear, however, is that this is the {\bf only possible}  Automorphic Lie Algebra.
The rest of the paper is intended to show that this is indeed the case.

For later use (Cf. Remark \ref{rem:undressing}) we mention that \(\det( \rho(\tf{A}{\pm}{}))=0\)
and \(\det(\rho(\tf{A}{0}{}))=-\omega_\alpha^2\alpha^2\).
\end{Remark}

\section{$\sl$--Automorphic Lie Algebras associated to finite groups}\label{sec:alias}  
Let $G\subset \aut (\bbbc(\lambda)\otimes\tf{g}{}{})$. A Lie  algebra 
$(\bbbc(\lambda)\otimes\tf{g}{}{})^G$ is called \emph{automorphic}, if its elements  $a\in(\bbbc(\lambda)\otimes\tf{g}{}{})^G$ are invariant $g(a)=a$ with respect to all automorphisms $g\in G$, i.e.
$$(\bbbc(\lambda)\otimes\tf{g}{}{})^G=\{ a\in(\bbbc(\lambda)\otimes\tf{g}{}{})\, \,|\,\, \phi (a)=a\, ,\, \forall 
\phi\in G\subset \aut (\bbbc(\lambda)\otimes\tf{g}{}{})\}\,.$$
We consider \(G\) to be a finite group; in particular, let \(G\) be one of the groups in the list (\ref{eq:list}); we aim for a complete classification for the case \(\tf{g}{}{}=\sl\) using geometric data.
This leads us to sketch a classification programme for Automorphic Lie Algebras associated to finite groups more in general, that is beyond \(\sl\). A key feature of this approach is the study of these algebras in the context of classical invariant theory. Indeed, the problem of reduction can be formulated in a uniform way using the theory of invariants. 
We consider first the problem of invariants starting from \(\bbbc[X,Y]\); through a homogenisation we will then map it to \(\bbbc(\lambda)\), where \(\lambda=X/Y\) (or  \(\lambda=Y/X\)).

\newcommand\mm[4]{\left(\begin{array}{cc}#1&#2\\#3&#4\end{array}\right)}
\subsection{The trace-form}\label{sec:traceform}
Given a representation \(\rho\) of \(\sl\), we can define the {\bf trace-form} \(\langle \cdot,\cdot\rangle\) by
\[
\langle \mf{X},\mf{Y}\rangle =\tr (\rho(\mf{X})\rho(\mf{Y})).
\]
\begin{Lemma}
Let \(\rho\) be the standard representation in \(\gl\) and \(f,g \in \bbbc[X,Y]\) with \(\omega_{f,g}\geq 2\).
Then \(\rho(\mf{A})=\mm{XY}{-X^2}{Y^2}{-XY}\) and
\begin{enumerate}
\item \(\langle \mf{A},\mf{A}\rangle=0\).
\item \(\langle \mf{A},\tf{A}{f}{1}\rangle=0\).
\item \(\langle \tf{A}{f}{1},\tf{A}{g}{1}\rangle =2\omega_f\omega_g f g\).
\item \(\langle \mf{A}, \tf{A}{f}{2}\rangle=-4\binom{\omega_f}{2}f\).
\item \(\langle\tf{A}{f}{1},\tf{A}{g}{2}\rangle=-4(\omega_g-1)(f,g)^1\).\label{item:trans}
\item \(\langle\tf{A}{f}{2},\tf{A}{g}{2}\rangle=-4(f,g)^2\).\label{item:trans2}
\end{enumerate}
\end{Lemma}
\begin{proof}
We prove here items \ref{item:trans} and \ref{item:trans2} and leave the other relations to the reader.
\begin{eqnarray*}
\langle\tf{A}{f}{1},\tf{A}{g}{2}\rangle&=&2\,\tr(\mm{- Xf_X+Yf_Y}{-2Xf_Y}{-2Yf_X}{- Yf_Y+ Xf_X}\mm{-g_{XY}}{-g_{YY}}{g_{XX}}{g_{XY}})
\\&=&4(-f_Y(Yg_{XY}+Xg_{XX})+f_X(Xg_{XY}+Yg_{YY}))
\\&=&4(\omega_g-1)(-f_Yg_{X}+f_Xg_{Y})
\\&=&-4(\omega_g-1)(f,g)^1\,,
\end{eqnarray*}
and
\begin{eqnarray*}
\langle\tf{A}{f}{2},\tf{A}{g}{2}\rangle&=&4\tr(\mm{-f_{XY}}{-f_{YY}}{f_{XX}}{f_{XY}}\mm{-g_{XY}}{-g_{YY}}{g_{XX}}{g_{XY}})
\\&=&4 (-f_{YY} g_{XX} + 2 f_{XY} g_{XY} - f_{XX} g_{YY})
\\&=&-4(f,g)^2\,.
\end{eqnarray*}
Observe that apparently \(|\langle \mf{A},\mf{B}\rangle|=|\mf{A}|+|\mf{B}|\mod 2\).
\end{proof}
\begin{Example}\label{D_2ex2}
Consider the same setting as in Example \ref{D_2ex1}, namely \(\mf{A}\)  be the invariant form in Theorem \ref{theo:formA}  and  \(\alpha=XY\); it follows then that all trace-forms \(\langle\tf{A}{\alpha}{i},\tf{A}{\alpha}{j}\rangle\) vanish with the only exception of 
\begin{eqnarray*}
\langle\tf{A}{\alpha}{1},\tf{A}{\alpha}{1}\rangle&=& 
\tr(\mm{0}{- 2X^2}{-2Y^2}{0}\mm{0}{-2X^2}{-2Y^2}{0})
\\&=&4X^2Y^2+4X^2Y^2
\\&=&2\,\omega_{\alpha}^{2}\,\alpha^2\,,\\
\langle\tf{A}{\alpha}{2},\tf{A}{\alpha}{2}\rangle&=& 
\tr(\mm{-2}{0}{0}{2}\mm{-2}{0}{0}{2})
\\&=&4+4
\\&=&-4(\alpha,\alpha)^2\,.
\end{eqnarray*}
\end{Example}

\subsection{Stanley and Clebsch-Gordan decompositions}\label{decomposition}
An essential step in the construction of the algebra \(\htf{g}{G}{}\,\)  is to find a basis for the covariant matrices.
This is done in this section by tensoring the Stanley basis of the covariant polynomials with
the selfadjoint representation of \(\sl\).
\begin{Definition}[Spherical polynomial rings]
Let \(\sigma\) be a faithful projective representation of a finite group \(G\) on \(\bbbc^2\). 
Let  \(\alpha,\beta\in\bbbc[X,Y]_G\) and \(\gamma=(\alpha,\beta)^1\)
and assume that every covariant can be written as an element in \(\bbbc[\alpha,\beta,\gamma]\).
As before we consider \(\alpha\) and \(\beta\) as the even elements,
and \(\gamma\) as the odd one (this supposes that \(\gamma\notin\bbbc[\alpha,\beta]\)).
We assume \((\alpha,\gamma)^1\in\bbbc[\alpha,\beta]\) and \((\beta,\gamma)^1\in\bbbc[\alpha,\beta]\)
and we notice that
\[
\omega_\gamma\gamma^2=
\omega_\beta\beta(\alpha,\gamma)^1
-\omega_\alpha\alpha(\beta,\gamma)^1
\in\bbbc[\alpha,\beta]
\]
This implies that the ring of covariants has the following Stanley decomposition:
\begin{equation}\label{eq:stanleydec}
\bbbc[X,Y]_G=\bbbc[\alpha,\beta]\oplus \bbbc[\alpha,\beta]\gamma.
\end{equation}
We call such a ring a {\bf spherical polynomial ring}.
\end{Definition}
\begin{Remark}
For \(\bbbi, \bbbo\) and \(\bbbt\), the \(\Zn{2}\)-grading of \(\beta\) is automatically \(0\).
In the case of \(\bbbd_m\), where \(\beta\) is not the second transvectant of \(\alpha\),
we put \(|\beta|=0\).
In both cases the  \(\Zn{2}\)-grading of \(\gamma\) is \(1\).
\end{Remark} 
\begin{Theorem}
The faithful projective representation of the groups \(\bbbt\), \(\bbbo\), \(\bbbi\) and \(\bbbd_m\) give rise to spherical polynomial rings
of covariants. The Molien functions are given in \cite{MR909219}.
\end{Theorem}
\begin{Remark}
The group \(\Zn{n}\) does not fit in since \(\alpha=X\), \(\beta=Y\), and  so \(\gamma\) is constant.
The ring of covariants is \(\bbbc[\alpha,\beta]\).
The covariant Lie algebra is spanned by \(\tf{A}{\alpha^2}{2}, \tf{A}{\alpha\beta}{2}\) and \(\tf{A}{\beta^2}{2}\), that is, by the orginal \(\sl\).
The details are given in Appendix \ref{AppZN}.
\end{Remark}
\begin{Definition}\label{def:grading}
We put a \(\Zn{2}\)-grading on 
\[(\bbbc[\alpha,\beta]\oplus \bbbc[\alpha,\beta](\alpha,\beta)^1)\otimes \mf{A}
\]
by setting \( |\alpha|=0\) (\(|\beta|=0\) in the case of \(\bbbd_m\)), \( |\tf{A}{}{}|=1\), \(|(f,g)^j|=|f|+|g|+j\mod 2\).
\end{Definition}
\begin{Theorem}\label{GCTHEO}
Let \(\mf{A}\) be an invariant quadratic form with coefficients in a \(G\)-module and assume \(\alpha,\beta\in\bbbc[X,Y]_G\) with \(\omega_{\alpha,\beta}\geq 2\)
and \(\gamma=(\alpha,\beta)^1\notin\bbbc[\alpha,\beta]\).
Then it follows from the Clebsch-Gordan decomposition theorem (see \cite{SVM}) that
\begin{equation}\label{eq:CG}
(\bbbc[\alpha,\beta]\oplus \bbbc[\alpha,\beta]\gamma)\otimes \sl
= \bbbc[\alpha,\beta]\otimes(\mf{A}
\oplus \tf{A}{\alpha}{2}
\oplus \tf{A}{\alpha}{1}
\oplus\tf{A}{\beta}{2}
\oplus\tf{A}{\beta}{1}
\oplus \tf{A}{\gamma}{2}).
\end{equation}
\end{Theorem}
\newcommand{\tp}[1]{t_{#1}^{\omega_{#1}}}
\newcommand{\tpv}[2]{t_{#1}^{#2}}
\begin{proof}
We compute the generating function of the left and right hand side.
Observe that \(\tp{f}\otimes\tpv{\mf{A}}{2}=\tpv{\tf{A}{f}{2}}{\omega_f-2}\oplus \tpv{\tf{A}{f}{1}}{\omega_f}\oplus\tp{f}\tpv{\mf{A}}{2}\),
if \(\omega_f\geq 2\).
Here we put a subindex on \(t\), to show where it comes from. Strictly speaking this is not necessary,
but it makes the identification of the remaining terms with the basis easier
and, instead of just giving a counting proof, it shows how the result  is obtained.
\begin{eqnarray*}
\lefteqn{\frac{1+\tp{\gamma}}{(1-\tp{\alpha})(1-\tp{\beta})}\otimes\tpv{\mf{A}}{2}=}&&
\\&=&\left(1+\frac{\tp{\beta}}{1-\tp{\beta}}+\frac{\tp{\alpha}}{(1-\tp{\alpha})(1-\tp{\beta})}+\frac{\tp{\gamma}}{(1-\tp{\alpha})(1-\tp{\beta})}\right)\otimes\tpv{\mf{A}}{2}
\\&=&\frac{1+\tp{\gamma}}{(1-\tp{\alpha})(1-\tp{\beta})}\tpv{\mf{A}}{2}
+\left(\frac{\tpv{\tf{A}{\beta}{1}}{\omega_\beta}}{1-\tp{\beta}}+\frac{\tpv{\tf{A}{\alpha}{1}}{\omega_\alpha}}{(1-\tp{\alpha})(1-\tp{\beta})}+\frac{\tpv{\tf{A}{\gamma}{1}}{\omega_\gamma}}{(1-\tp{\alpha})(1-\tp{\beta})}\right)
\\&+&\left(\frac{\tpv{\tf{A}{\beta}{2}}{\omega_\beta-2}}{1-\tp{\beta}}+\frac{\tpv{\tf{A}{\alpha}{2}}{\omega_\alpha-2}}{(1-\tp{\alpha})(1-\tp{\beta})}+\frac{\tpv{\tf{A}{\gamma}{2}}{\omega_\gamma-2}}{(1-\tp{\alpha})(1-\tp{\beta})}\right)
\\&=&\frac{\tpv{\mf{A}}{2}+\tp{\gamma}\tpv{\mf{A}}{2}+\tpv{\tf{A}{\beta}{1}}{\omega_\beta}-\tp{\alpha}\tpv{\tf{A}{\beta}{1}}{\omega_\beta}
+\tpv{\tf{A}{\alpha}{1}}{\omega_\alpha}+\tpv{\tf{A}{\gamma}{1}}{\omega_\gamma}+\tpv{\tf{A}{\beta}{2}}{\omega_\beta-2}
-\tp{\alpha}\tpv{\tf{A}{\beta}{2}}{\omega_\beta-2}+\tpv{\tf{A}{\alpha}{2}}{\omega_\alpha-2}+\tpv{\tf{A}{\gamma}{2}}{\omega_\gamma-2}}{(1-\tp{\alpha})(1-\tp{\beta})}
\\&=&\frac{\tpv{\mf{A}}{2}+\tpv{\tf{A}{\beta}{1}}{\omega_\beta}
+\tpv{\tf{A}{\alpha}{1}}{\omega_\alpha}+\tpv{\tf{A}{\beta}{2}}{\omega_\beta-2}
+\tpv{\tf{A}{\alpha}{2}}{\omega_\alpha-2}+\tpv{\tf{A}{\gamma}{2}}{\omega_\gamma-2}}{(1-\tp{\alpha})(1-\tp{\beta})}
\\&=&\frac{(1+t^{\omega_\alpha-2})(\tpv{\mf{A}}{2}
 +\tpv{\tf{A}{\beta}{2}}{\omega_\beta-2}
 +\tpv{\tf{A}{\alpha}{2}}{\omega_\alpha-2})}{(1-\tp{\alpha})(1-\tp{\beta})}\,.
\end{eqnarray*}
\begin{Remark}
The factorization in  the last step seems to be connected to Corollary \ref{cor:subalgebra}.
See also \cite[$\mathsection 8$, On the reflection representation of certain Hecke algebras]{MR909219}.
\end{Remark}
The underlying relations which are used to get rid of the minus signs are:
\begin{eqnarray*}
2\gamma\mf{A}&=&\omega_\beta\,\beta\,\tf{A}{\alpha}{1}-\omega_\alpha\,\alpha\,\tf{A}{\beta}{1}\,,
\\
\frac{(\omega_\alpha-1)(\omega_\beta-1)}{\omega_\gamma}\tf{A}{\gamma}{1}&=&
\frac{\omega_\alpha-\omega_\beta}{\omega_\gamma}(\alpha,\beta)^2\mf{A}+\binom{\omega_\beta}{2}\beta\,\tf{A}{\alpha}{2}-\binom{\omega_\alpha}{2}\alpha\,\tf{A}{\beta}{2}\,,
\end{eqnarray*}
where we know that \((\alpha,\beta)^2\mf{A}\in\bbbc[\alpha,\beta]\mf{A}\oplus \bbbc[\alpha,\beta]\gamma\mf{A}
\subset \bbbc[\alpha,\beta](\mf{A}\oplus \tf{A}{\alpha}{1}\oplus\tf{A}{\beta}{1})\).
Now that the counting is in order, it suffices to show that the right hand side is indeed a direct sum.
To this end assume that
\[
F_1(\alpha,\beta)\mf{A}
+F_2(\alpha,\beta)\tf{A}{\alpha}{2}
+F_3(\alpha,\beta)\tf{A}{\alpha}{1}
+F_4(\alpha,\beta)\tf{A}{\beta}{2}
+F_5(\alpha,\beta)\tf{A}{\beta}{1}
+F_6(\alpha,\beta)\tf{A}{\gamma}{2}=0.
\]
Then, taking the trace-form with \(\mf{A}\), we obtain
\begin{eqnarray*}
0&=&
\langle \mf{A}, F_2(\alpha,\beta)\tf{A}{\alpha}{2}\rangle
+\langle \mf{A}, F_4(\alpha,\beta)\tf{A}{\beta}{2}\rangle
+\langle \mf{A}, F_6(\alpha,\beta)\tf{A}{\gamma}{2}\rangle
\\&=&
-4\binom{\omega_\alpha}{2} \alpha F_2(\alpha,\beta)
-4\binom{\omega_\beta}{2} \beta F_4(\alpha,\beta)
-4\binom{\omega_\gamma}{2} \gamma F_6(\alpha,\beta)\,.
\end{eqnarray*}
This implies \(F_6=0\) and \(\binom{\omega_\alpha}{2} \alpha F_2(\alpha,\beta)
+\binom{\omega_\beta}{2} \beta F_4(\alpha,\beta)=0\).
Next we take the trace-form with \(\tf{A}{\alpha}{1}\) to obtain
\begin{eqnarray*}
0&=&
\langle \tf{A}{\alpha}{1},F_3(\alpha,\beta)\tf{A}{\alpha}{1}\rangle
+\langle \tf{A}{\alpha}{1},F_4(\alpha,\beta)\tf{A}{\beta}{2}\rangle
+\langle \tf{A}{\alpha}{1},F_5(\alpha,\beta)\tf{A}{\beta}{1}\rangle
\\&=& 2\omega_\alpha^2 \alpha^2 F_3(\alpha,\beta)
-4(\omega_\beta-1) \gamma F_4(\alpha,\beta)
+2\omega_\alpha\omega_\beta\alpha\beta F_5(\alpha,\beta)\,.
\end{eqnarray*}
This implies \(F_4=0\) (and therefore \(F_2=0\)) and \( \omega_\alpha \alpha F_3(\alpha,\beta)
+\omega_\beta\beta F_5(\alpha,\beta)=0\).
Finally, taking the trace-form with \(\tf{A}{\alpha}{2}\), we obtain
\begin{eqnarray*}
0&=&
\langle F_1(\alpha,\beta)\mf{A},\tf{A}{\alpha}{2}\rangle
+\langle F_5(\alpha,\beta)\tf{A}{\beta}{1},\tf{A}{\alpha}{2}\rangle
\\&=&-4\binom{\omega_\alpha}{2}\alpha F_1(\alpha,\beta)
+4(\omega_\beta-1)\gamma F_5(\alpha,\beta).
\end{eqnarray*}
This implies \(F_5=0=F_1\), and therefore \(F_3=0\).
This shows that indeed the sum is direct, as claimed.
This concludes the \(SL(2,\bbbc)\) part of the proof.

Finally the Molien function for each group \(\bbbi, \bbbo, \bbbt\) and \(\bbbd_m\)
for the given matrix representation has been computed and it coincides with our result in all cases.
\end{proof}
The Lie algebra generated by \( \mf{A}\), \(\tf{A}{\alpha}{2}\),\(\tf{A}{\alpha}{1}\), \(\tf{A}{\beta}{2}\), \(\tf{A}{\beta}{1}\), \(\tf{A}{(\alpha,\beta)^1}{2}\) with coefficient ring \(\bbbc[\alpha,\beta]\) is the algebra of covariants \(\htf{g}{G}{}\,\), where we recall that \(\htf{g}{}{}=\mb{C}[X,Y]\otimes\tf{g}{}{}\).
\begin{Corollary}\label{cor:subalgebra}
It follows from the grading that
\begin{equation}
((\bbbc[\alpha,\beta]\oplus \bbbc[\alpha,\beta](\alpha,\beta)^1)\otimes \mf{A})_0
= \bbbc[\alpha,\beta]\otimes(
\tf{A}{\alpha}{1}
\oplus\tf{A}{\beta}{1}
\oplus \tf{A}{\gamma}{2})
\end{equation}
and
\begin{equation}
((\bbbc[\alpha,\beta]\oplus \bbbc[\alpha,\beta](\alpha,\beta)^1)\otimes \mf{A})_1
= \bbbc[\alpha,\beta]\otimes(\mf{A}
\oplus \tf{A}{\alpha}{2}
\oplus\tf{A}{\beta}{2})\,.
\end{equation}
It also follows from Theorem \ref{Theorem:THEOGR}  that \(\bbbc[\alpha,\beta]\otimes(
\tf{A}{\alpha}{1}
\oplus\tf{A}{\beta}{1}
\oplus \tf{A}{\gamma}{2})\) is a Lie subalgebra.
We find later that these are the only elements that can be mapped to invariant homogeneous elements with divisors \(\alpha\) or \(\beta\)
in the Automorphic Lie Algebra (see Section \ref{sec:homo}).
\end{Corollary}

\subsection{Description of the group actions}\label{sec:groups}
We describe the group action on \(\bbbc^2\) for later use. In those cases where the action is not given explicitly,
it is given implicitly, since it permutes the zeros of the groundforms. Notice that the degrees of the groundforms 
correspond with number of vertices, edges and faces in the cases of \(\bbbi, \bbbo\) and \(\bbbt\).

\subsubsection{Icosahedral group  \(\bbbi\)}\label{sec:i}
An icosahedron is a convex regular polyhedron (a Platonic solid) with twenty triangular faces, thirty edges and twelve vertices. 
A regular icosahedron has $60$ rotational (or orientation-preserving) symmetries; the set of orientation-preserving symmetries forms a group referred to as $\bbbi$; $\bbbi$ is isomorphic to $A_5$, the alternating group of even permutations of five objects. As an abstract group it is generated by two elements, \(s\) and \(r\), satisfying the identities \[s^{5}=r^{2}=id\,, \quad\,(s\,r)^3=id\,.\]

The \(\bbbi\)--groundforms are given by \(\alpha=XY (X^{10}+11 X^{5} Y^{5}-Y^{10})\), \(\beta=(\alpha,\alpha)^2\) and \(\gamma=(\alpha,\beta)^1\). Since the degree of \(\alpha\), \(\omega_\alpha\), is equal to \(12\) it follows that \(\omega_\beta=2\,\omega_\alpha-4=20\) and \(\omega_\gamma=3\,\omega_\alpha-6=30\). 

\subsubsection{Octahedral group \(\bbbo\)}\label{sec:o}
A regular octahedron is a Platonic solid composed of eight equilateral triangles, four of which meet at each vertex; it has six vertices and eight edges.
A regular octahedron has 24 rotational (or orientation-preserving) symmetries. A cube has the same set of symmetries, since it is its dual. The group of orientation-preserving symmetries is denoted by \(\bbbo\) and it is isomorphic to \(S_4\), or the group of permutations of four objects, since there is exactly one such symmetry for each permutation of the four pairs of opposite sides of the octahedron. As an abstract group it is generated by two elements, \(s\) and \(r\), satisfying the identities \[s^{4}=r^{2}=id\,, \quad\,(s\,r)^3=id\,.\]

The classical \(\bbbo\)--groundforms are given by \(\alpha=XY (X^{4}-Y^{4})\), \(\beta=(\alpha,\alpha)^2\) and \(\gamma=(\alpha,\beta)^1\). Since \(\omega_\alpha=6\) it follows that \(\omega_\beta=2\,\omega_\alpha-4=8\) and \(\omega_\gamma=3\,\omega_\alpha-6=12\). 

\subsubsection{Tetrahedral group \(\bbbt\)}\label{sec:t}
A regular tetrahedron is a regular polyhedron composed of four equilateral triangular faces, three of which meet at each vertex. It has four vertices and six edges. A regular tetrahedron is a Platonic solid; it has 12 rotational (or orientation-preserving) symmetries; the set of orientation-preserving symmetries forms a group referred to as \(\bbbt\), isomorphic to the alternating subgroup \(A_4\). As an abstract group it is generated by two elements, \(s\) and \(r\), satisfying the identities \[s^{3}=r^{2}=id\,, \quad\,(s\,r)^3=id\,.\]

The \(\bbbt\)--groundforms are given by \(\alpha=X^{4}-2i\sqrt{3}X^2 Y^2+Y^{4}\), \(\beta=(\alpha,\alpha)^2\) and \(\gamma=(\alpha,\beta)^1\). Since \(\omega_\alpha=4\) it follows that \(\omega_\beta=2\,\omega_\alpha-4=4\) and \(\omega_\gamma=3\,\omega_\alpha-6=6\). 

\subsubsection{Dihedral group \(\bbbd_n\)}\label{sec:DN}
Let us now turn our attention to the symmetry group of a dihedron, the dihedral group; this is fundamentally different from the previous cases, since \(\beta\neq(\alpha,\alpha)^2\) and it has to be computed independently; however the rest of the procedure is the same.

As in the case of \(\Zn{n}\) (see Appendix \ref{AppZN}), in order to get an action of \(\bbbd_n\) on the spectral parameter \(\lambda\)
we need to act with \(\bbbd_{m}\) on the \(X,Y\)-plane, where \(m=n\) if \(n\) is odd, and \(m=2n\) when \(n\) is even.
The dihedral group \(\bbbd_{m}\) is the group of rotations and reflections of the plane which preserve a regular polygon with \(m\) vertices. It is generated by two elements, \(s\) and \(r\), satisfying the identities 
\[s^{m}=r^{2}=id\,, \quad\,r\,s\,r=s^{-1}\,.\]

We take a projective representation of the group,
defined by \(\sigma(s)=\left(\begin{array}{cc} \omega & 0\\ 0 &\omega^{m-1}\end{array}\right)\),
with \(\omega\) an elementary \(m\)th root of unity,
and
\(\sigma(r)=\left(\begin{array}{cc} 0 & i\\ i & 0 \end{array}\right)\). 
Let \(\alpha=XY\), \(\beta=\frac{1}{2}(X^m+Y^m)\) and  \(\gamma=(\alpha,\beta)^1\), i.e.
\[
\gamma=(\alpha,\beta)^1=\alpha_Y\beta_X-\alpha_X\beta_Y= \frac{m}{2}(X^m-Y^m)\,,
\]
that is 
\(\omega_\alpha=2\), \(\omega_\beta=m\), \(\omega_\gamma=2+m-2=m\) and one has \(\gamma^2-m^2\beta^2=- m^2 \alpha^m\).
Then \(\sigma^\star(s)\alpha=\alpha\), \(\sigma^\star(r)\alpha=i^2\alpha\),
\(\sigma^\star(s)\beta=\beta\), \(\sigma^\star(r)\beta=i^m\beta\). The action
of \(\gamma\) is given by the product of \(\alpha\) and \(\beta\).
\begin{Example}[$\bbbd_n$, with $n=2$]\label{D_2ex3} 
Let \(\alpha=XY\) and let \(\mf{A}\)  be the invariant form given in Theorem \ref{theo:formA}; recall also Example \ref{D_2ex1}; one has 
\begin{eqnarray*}
\rho(\tf{A}{\alpha}{1})
&=& \left(
\begin{array}{cc}
0 &-2\,X^2\\
-2\,Y^2 & 0
\end{array}
\right),\\
\rho(\tf{A}{\beta}{1})
&=& \left(
\begin{array}{cc}
2Y^4- 2X^4 &-4XY^3\\
-4X^3Y &  2X^4- 2Y^4 
\end{array}
\right),\\
\rho(\tf{A}{\gamma}{2})
&=& \left(
\begin{array}{cc}
0 & 48 Y^2\\
48 X^2 & 0
\end{array}
\right).
\end{eqnarray*}
\end{Example}

\subsection{Structure constants}\label{sec:structure}
In this Section we use the transvectant formula and the trace-form to derive the commutation relations of the even part of the covariant algebra \(\htf{g}{G}{}\,\).
We first present a few technical Lemmas, to be proven by checking them for each group.
\begin{Lemma}
For the groups \(\bbbt, \bbbo, \bbbi\) and \(\bbbd_m\), \(\tf{p}{G}{}(\beta)=\frac{(\alpha,\gamma)^1}{\beta}\in\bbbc[\beta]\).
\end{Lemma}
\begin{proof}
See Table \ref{table_pbeta}.
\end{proof}
\begin{center}
\begin{table}[ht!]
\begin{center}
\begin{tabular}{|c|c|c|c|c|} \hline
\(G\) & \(\bbbi\)  &  \(\bbbo\) & \(\bbbt\)  & \(\bbbd_m\) \\ 
\hline \hline 
\(\tf{p}{G}{}(\beta)\) & \(-\frac{300}{121}\,\beta\)&\(-\frac{48}{25}\,\beta\) & \(-\frac{4}{3}\,\beta\)  & \(m^2\)\\ 
\hline
\end{tabular}
\end{center}
\caption{$\tf{p}{G}{}(\beta)$ for the groups \(\bbbi, \bbbo, \bbbt\) and \(\bbbd_m\)}
\label{table_pbeta}
\end{table}
\end{center}
\begin{Lemma}
For the groups \(\bbbi, \bbbo, \bbbt\) and \(\bbbd_m\), \(\tf{q}{G}{}(\alpha)=\frac{(\beta,\gamma)^1}{\alpha}\in\bbbc[\alpha]\).
\end{Lemma}
\begin{proof}
See Table \ref{table_qalpha}.
\end{proof}
\begin{center}
\begin{table}[ht!]
\begin{center}
\begin{tabular}{|c|c|c|c|c|} \hline
 \(G\)& \(\bbbi\)  &  \(\bbbo\) & \(\bbbt\)  & \(\bbbd_m\) \\ 
\hline \hline 
\(\tf{q}{G}{}(\alpha)\) &  \(-101198592000\,\alpha^3\)  & \(34560000\,\alpha^2\) & \(-3538944 \,i\,\sqrt{3}\,\alpha\) & \(\frac{1}{2}m^3\,\alpha^{m-2}\)\\ 
\hline
\end{tabular}
\end{center}
\caption{$\tf{q}{G}{}(\alpha)$ for the groups \(\bbbi, \bbbo, \bbbt\) and \(\bbbd_m\)}
\label{table_qalpha}
\end{table}
\end{center}
\begin{Corollary}\label{gamma2}
\( \omega_\gamma\gamma^2=-\omega_\alpha\alpha^2\tf{q}{G}{}(\alpha)+\omega_\beta\beta^2\tf{p}{G}{}(\beta) \).
\end{Corollary}
\begin{Lemma}\label{lem:gamma2}
For the groups \(\bbbt, \bbbo, \bbbi\) and \(\bbbd_m\), 
\begin{equation}\label{gammagamma2}
2 (\omega_\gamma-1)^2\tf{p}{G}{}(\beta)\tf{q}{G}{}(\alpha)+\omega_\alpha\,\omega_\beta(\gamma,\gamma)^2=0\,.
\end{equation}
\end{Lemma}
\begin{proof}
Observe first that 
\begin{center}
\begin{table}[ht!]
\begin{center}
\begin{tabular}{|c|c|c|c|c|} \hline
 \(G\)& \(\bbbi\)  &  \(\bbbo\) & \(\bbbt\)  & \(\bbbd_m\) \\ 
\hline \hline 
\((\gamma,\gamma)^2\) & \(-1758430080000\,\alpha^3\,\beta\)  & \(334540800\,\alpha^2\,\beta\) & \(-14745600 \,i\,\sqrt{3}\,\alpha\,\beta\)  & \(-\frac{1}{2}m^4(m-1)^2\,\alpha^{m-2}\)\\ 
\hline
\end{tabular}
\end{center}
\caption{\((\gamma,\gamma)^2\) for the groups \(\bbbt, \bbbo, \bbbi\) and \(\bbbd_m\)}
\label{table_gammagamma2}
\end{table}
\end{center}
Substituting these relations and those of Tables \ref{table_pbeta} and \ref{table_qalpha} into (\ref{gammagamma2}) proves the lemma.
\end{proof}
The following Lemma will be used later in Theorem \ref{theo:big}.
\begin{Lemma}\label{lem:lem45}
\(2(\omega_\alpha-1)^2\tf{p}{G}{}(\beta)=-\omega_\beta\omega_\gamma(\alpha,\alpha)^2\).
\end{Lemma}
\begin{proof}
By inspection of Table \ref{table_pbeta}.
\end{proof}
We are now ready to derive the commutation relations of the even part of the covariant algebra \(\htf{g}{G}{}\,\):
\begin{Theorem}\label{THEO:StuctureConstants}
The commutation relations on \(\bbbc[\alpha,\beta]\otimes(\tf{A}{\alpha}{1}\oplus\tf{A}{\beta}{1}\oplus \tf{A}{\gamma}{2})\) are
\begin{eqnarray}
{[\tf{A}{\alpha}{1},\tf{A}{\beta}{1}]}&=&2\omega_\beta\beta\tf{A}{\alpha}{1}-2\omega_\alpha\alpha\tf{A}{\beta}{1},\label{eq:commab}\\
{[\tf{A}{\alpha}{1},\tf{A}{\gamma}{2}]}&=&4(\omega_\gamma-1)\frac{\tf{p}{G}{}(\beta)}{\omega_\beta}\tf{A}{\beta}{1}+2\omega_\alpha\alpha\tf{A}{\gamma}{2}\label{eq:commac},\\
{[\tf{A}{\beta}{1},\tf{A}{\gamma}{2}]}&=&4(\omega_\gamma-1)\frac{\tf{q}{G}{}(\alpha)}{\omega_\alpha}\tf{A}{\alpha}{1}+2\omega_\beta\beta\tf{A}{\gamma}{2}\label{eq:commbc}.
\end{eqnarray}
\end{Theorem}
\begin{proof}
For the the first commutation relation (\ref{eq:commab}) one has
\begin{eqnarray*}
\lefteqn{[( \alpha,\mf{A})^1,(\beta,\mf{A})^1]=}&&
\\&=&X^{-2}[2\alpha_Y\mf{A}-\omega_\alpha\alpha\gf{A}{Y},2\beta_Y\mf{A}-\omega_\beta \beta\gf{A}{Y}]
\\&=&
X^{-2}(-2\omega_\beta\alpha_Y\beta[\mf{A},\gf{A}{Y}] -2\omega_\alpha\alpha \beta_Y[\gf{A}{Y},\mf{A}])
\\&=&-2X^{-1}(-2\omega_\beta\alpha_Y\beta+2\omega_\alpha\alpha \beta_Y)\mf{A}
\\&=&2\omega_\beta\, \beta\,\tf{A}{\alpha}{1}-2\omega_\alpha\,\alpha\,\tf{A}{\beta}{1}\,.
\end{eqnarray*}
It follows from this relation that
\begin{eqnarray}
{\langle [\tf{A}{\alpha}{1},\tf{A}{\beta}{1}],\tf{A}{\gamma}{2}\rangle}&=&-8\omega_\beta(\omega_\gamma-1)\beta(\alpha,\gamma)^1+8\omega_\alpha(\omega_\gamma-1)\alpha(\beta,\gamma)^1\nonumber 
\\&=&8\omega_\gamma(\omega_\gamma-1)\gamma^2.\label{eq:commaba}
\end{eqnarray}
Consider next the matrix of the trace-forms of the basis \(\tf{A}{\alpha}{1}, \tf{A}{\beta}{1}\) and 
\(\tf{A}{\gamma}{2}\):
\[
\left(\begin{array}{ccc}2 \omega_\alpha^2\alpha^2&2\omega_\alpha\omega_\beta \alpha\beta&-4(\omega_\gamma-1) (\alpha,\gamma)^1\\
2\omega_\alpha\omega_\beta \alpha\beta&2\omega_\beta^2\beta^2&-4(\omega_\gamma-1)(\beta,\gamma)^1\\
-4(\omega_\gamma-1)(\alpha,\gamma)^1 &-4(\omega_\gamma-1)(\beta,\gamma)^1 &4 (\gamma,\gamma)^2
\end{array}\right)\,;
\]
its determinant is \(-32(\omega_\gamma-1)^2\omega_\gamma^2\gamma^4\).
This implies that the trace-form is nondegenerate and that we can prove the commutation relations
by taking the trace-form with the basis elements. When this results in identities, we have a proof.

Taking the trace-form of (\ref{eq:commac}) with \(\tf{A}{\alpha}{1}\) we see that it trivialises (i.e. it is equal to zero):
\begin{eqnarray*}
0&=&4(\omega_\gamma-1)\frac{(\alpha,\gamma)^1}{\omega_\beta\beta}\langle \tf{A}{\alpha}{1},\tf{A}{\beta}{1}\rangle+2\omega_\alpha\alpha\langle\tf{A}{\alpha}{1},\tf{A}{\gamma}{2}\rangle
\\&=&8\omega_\alpha(\omega_\gamma-1)(\alpha,\gamma)^1\alpha-8\omega_\alpha(\omega_\gamma-1)\alpha(\alpha,\gamma)^1\,.
\end{eqnarray*}
Taking the trace-form of (\ref{eq:commac}) with \(\tf{A}{\beta}{1}\) results in the relation for \(\gamma^2\):
\begin{eqnarray*}
\langle\tf{A}{\beta}{1},{[\tf{A}{\alpha}{1},\tf{A}{\gamma}{2}]}\rangle&=&4(\omega_\gamma-1)\frac{(\alpha,\gamma)^1}{\omega_\beta\beta}\langle\tf{A}{\beta}{1},\tf{A}{\beta}{1}\rangle+2\omega_\alpha\alpha\langle\tf{A}{\beta}{1},\tf{A}{\gamma}{2}\rangle
\end{eqnarray*}
implies
\begin{eqnarray*}
8\omega_\gamma\,(\omega_\gamma-1)\,\gamma^2&=&8\omega_\beta(\omega_\gamma-1)(\alpha,\gamma)^1\beta-8(\omega_\gamma-1)\omega_\alpha\,\alpha\,(\beta,\gamma)^1.
\end{eqnarray*}
Taking the trace-form of (\ref{eq:commac}) with \(\tf{A}{\gamma}{2}\) results in 
\[
0= -2 (\omega_\gamma-1)^2(\alpha,\gamma)^1(\beta,\gamma)^1-\omega_\alpha\omega_\beta\,\alpha\,\beta\,(\gamma,\gamma)^2
\]
and this follows from Lemma \ref{lem:gamma2}.
The proof of (\ref{eq:commbc}) follows exactly the same pattern.
\end{proof}
\begin{Remark}
The commutation relations show that they only depend on  the relation among the groundforms
indicating that they are determined by the geometry of the curve given in Corollary \ref{gamma2}.
\end{Remark}
So far the covariant matrices are functions of \(X\) and \(Y\); in the following we define invariant homogeneous elements in the local coordinate \(\lambda=X/Y\) (or \(\lambda=Y/X\)).

\section{Towards Lax pairs: homogenisation}\label{sec:homo}
Suppose we are given an element \( M\in\htf{g}{G}{\chi}\), with \( M \) of degree \( k \) in \(X\) and \(Y\) and \(f\in \mb{C}[X,Y]_G^{\chi}\) of the same degree, we can now consider \(\frac{M}{f}\) as an element in \(\tf{g}{}{}\) with coefficients in the field of rational functions \(\bbbc(\lambda)\), where \(\lambda=X/Y\). 
To map \(\htf{g}{G}{}\,\) to the zero-homogeneous part of \(\bbbc(\alpha,\beta)\otimes(\tf{A}{\alpha}{1}\oplus\tf{A}{\beta}{1}\oplus \tf{A}{\gamma}{2})\) we proceed as follows.

\subsection{The groups $\bbbt$, $\bbbo$, $\bbbi$}
To obtain an homogeneous element we first of all need to choose an automorphic function to start with. This is equivalent to fixing the poles of our \(G\)--invariant elements. In other words, this choice corresponds to the choice of an orbit \(\Gamma\) of the group \(G\) (see \cite{lm_cmp05}). Let us choose \(\alpha\) (the \(\beta\)-choice is treated in Appendix \ref{sec:betadivall}) and consider
\[
\ttf{A}{f}{j}= \frac{\beta^{m_2}}{\alpha^{m_1}}\tf{A}{f}{j}\,\quad m_1\in\bbbz_{>0},\quad m_2\in\bbbz_{\geq 0}\,,
\]
where \(\tf{A}{1}{0}\equiv  \tf{A}{}{}\).
Therefore, for each \(\tf{A}{f}{j}\),  we need to solve the diophantine relation
\[
2-2j+\omega_f+2 m_2  ( \omega_\alpha-2)=m_1 \, \omega_\alpha\,,\quad m_1\in\bbbz_{>0},\quad m_2\in\bbbz_{\geq 0}\,,
\]
where we denote  \(\omega_{f=1}=0\) and we recall that  \(\omega_\beta=2\omega_\alpha-4\), \(\omega_\gamma=3\omega_\alpha-6\). This allows us to define
homogeneous elements, now in  \(\ttf{g}{}{G}\).

Suppose \(\omega_\alpha=0\mod 4\), as is the case for \(G=\bbbt, \bbbi\). Then for \(f=1,\,\alpha,\, \beta\) we have
\[
2|\tf{A}{f}{j}|=2-2j=0\mod 4\,,\quad m_1\in\bbbz_{>0},\quad m_2\in\bbbz_{\geq 0}\,,
\]
This implies \(|\tf{A}{f}{j}|=0\), excluding 
\(\tf{A}{}{}\), \(\tf{A}{\alpha}{2}\) and
\(\tf{A}{\beta}{2}\).

To these equations we add the invariance requirement.
On the covariants acts the abelianized group \(\mathcal{A}G=G/[G,G]\).
Since the icosahedral group is perfect,  \(\mathcal{A}\bbbi\) is trivial and so is its action.
Suppose \(\alpha\) goes to \(\chi\,\alpha\), then \(\beta\) goes to \(\chi^2\beta\) and \(\gamma\) to \(\chi^3\gamma\).
In the same manner \(\tf{A}{\alpha}{i}\) goes to \(\chi\,\tf{A}{\alpha}{i}\), \(i=1,2\),
\(\tf{A}{\beta}{i}\) goes to \(\chi^2\tf{A}{\beta}{i}\), \(i=1,2\) and \(\tf{A}{\gamma}{2}\) to \(\chi^3\tf{A}{\gamma}{2}\).
Furthermore \(\chi\) should be compatible with the relation 
\[ 
\omega_\gamma\gamma^2=-\omega_\alpha\alpha^2\tf{q}{}{}(\alpha)+\omega_\beta\beta^2\tf{p}{}{}(\alpha)\,,
\]
which implies \(\chi^4=\chi^{\deg(\tf{q}{G}{})}\).
It follows that \(\chi^3=1\) when \(G=\bbbt\), \( \chi^2=1\) 
when \(G=\bbbo\) and \(\chi=1\) when \(G=\bbbi\).
And this in turn implies that 
\begin{eqnarray*}
\bbbi:& \alpha\mapsto\alpha ,\quad\beta\mapsto\beta ,\quad\gamma\mapsto\gamma,\,\quad&\mathcal{A}\bbbi=1;\\
\bbbo:& \alpha\mapsto\chi\alpha ,\quad\beta\mapsto\beta ,\quad\gamma\mapsto\chi\gamma,\,\quad&\mathcal{A}\bbbo=\bbbz/2;\\
\bbbt:& \alpha\mapsto\chi\alpha ,\quad \beta\mapsto\chi^2\beta ,\quad\gamma\mapsto\gamma, \quad&\mathcal{A}\bbbt=\bbbz/3\,.\\
\end{eqnarray*}
This leads to
\begin{eqnarray*}
\bbbt&:& 2m_2 +\delta_f =m_1 \mod 3, \quad \delta_\alpha=1,\,\, \delta_\beta=2,\,\, \delta_\gamma=0\,;\\
\bbbo&:& \delta_f = m_1 \mod 2, \quad \delta_\alpha=1, \,\,\delta_\beta=0,\,\,\delta_\gamma=1\,.
\end{eqnarray*}
We now restrict to \(G=\bbbo\):
\begin{eqnarray}
\bbbo&:&1-j+\frac{\omega_f}{2}+4 m_2  =3 m_1  ,\quad m_1\in\bbbz_{>0},\quad m_2\in\bbbz_{\geq 0}\,,\\
\bbbo&:& \delta_f = m_1 \mod 2, \quad \delta_\alpha=1, \,\delta_\beta=0,\,\delta_\gamma=1\,.
\end{eqnarray}
Take \(f=\alpha\) and \(j=2\). Then \(m_1\) is even by the first equation and odd by the second.
This rules out \(\tf{A}{\alpha}{2}\).
With \(f=\beta\) and \(j=2\), \(m_1\) is odd by the first equation and even by the second. This rules out
\(\tf{A}{\beta}{2}\).
In the case \(j=0\) one has \(m_1\) is odd by the first and even by the second equation.
This rules out \(\tf{A}{}{}\).
\begin{Remark}
If one would have divided by \(\beta\) instead of \(\alpha\), the analysis leading to the excluded
elements remains the same.
\end{Remark}
\begin{Corollary}
The matrices \(\tf{A}{f}{j}\) such that \(|\tf{A}{f}{j}|=1\) cannot be homogenized,
neither by division by \(\alpha\) nor \(\beta\) powers.
\end{Corollary}
\subsection{The group $\bbbd_n$}

As before, we consider expressions \(\ttf{A}{f}{i}\) that are zero homogeneous in \(X,Y\)
and invariant under the \(\bbbd_n\) group action.
Let
\[
\ttf{A}{f}{j}=\frac{\beta^{m_2}\tf{A}{f}{j}}{\alpha^{m_1}}.
\]
Then the homogeneity equation is
\[
m\,m_2  + 2 - 2j +\omega_f = 2 m_1\,\quad m_2\geq 0,\quad m_1>0\,,
\]
and the invariance equation is
\[
m\, m_2 + \delta_f = 2 m_1 \mod 4
\]
where \(\delta_f=\omega_f\) for \(f=1,\,\alpha,\, \beta\), and \(\delta_\gamma=\omega_\beta+\omega_\alpha\). Recall that \(m=n\) if \(n\) is odd while \(m=2\,n\) if  \(n\) is even.

It follows that 
\[
2 - 2j +\omega_f = \delta_f \mod 4\,.
\]
Since for \(f=1,\,\alpha,\,\beta\), \(\,\delta_f=\omega_f\), this rules out \(\tf{A}{}{}\), \(\tf{A}{\alpha}{2}\) and \(\tf{A}{\beta}{2}\).

\subsection{The homogeneous basis}\label{sec:homobasis}
Having defined the homogeneous elements \(\ttf{A}{f}{i}\) in the previous Section we can now define a basis over the ring \(\bbbc(\ttf{I}{G}{})\), where \(\ttf{I}{G}{}\) is defined below:
\begin{Theorem}\label{theo:HomBasis}
The homogeneous basis of \(\ttf{\sl}{}{G}\), that is the \(G\)-Automorphic Lie Algebra based on \(\sl\) with poles in the zeros of \(\alpha\),
is given by
\begin{eqnarray}\label{homogeneous1}
\ttf{A}{\alpha}{1}&=&\frac{1}{\omega_\alpha\,\alpha}\tf{A}{\alpha}{1}, \\ \label{homogeneous2}
\ttf{A}{\beta}{1}&=&\frac{\ttf{I}{G}{}}{\omega_\beta\,\beta}\tf{A}{\beta}{1},\\ \label{homogeneous3}
\ttf{A}{\gamma}{2}&=&\frac{\omega_\alpha\alpha\,\ttf{I}{G}{}}{4(\omega_\gamma -1)\beta\,\tf{p}{G}{}(\beta)}\tf{A}{\gamma}{2},
\end{eqnarray}
where \(\ttf{I}{G}{}=\frac{\omega_\beta}{\omega_\alpha}\frac{\beta^2\tf{p}{G}{}(\beta)}{\alpha^2\tf{q}{G}{}(\alpha)} \) for
\(G=\bbbi, \bbbo, \bbbt\) and \(\bbbd_{n}\), \(n\) odd, and \(\ttf{I}{G}{2}=\frac{\omega_\beta}{\omega_\alpha}\frac{\beta^2\tf{p}{G}{}(\beta)}{\alpha^2\tf{q}{G}{}(\alpha)} \)
for \(G=\bbbd_{n}\), \(n\) even.
In other words, \(\frac{\omega_\beta}{\omega_\alpha}\frac{\beta^2\tf{p}{G}{}(\beta)}{\alpha^2\tf{q}{G}{}(\alpha)} = \ttf{I}{G}{1+\tf{r}{G}{}}\) where \(\tf{r}{G}{}=0\) for \(G=\bbbi, \bbbo, \bbbt,\bbbd_n\), \(n\) odd, and \(1\) for  \(\bbbd_{n}\), \(n\) even.
If we, moreover, define \(\ttf{J}{G}{}=\frac{\omega_\gamma\gamma^2}{\omega_\alpha\,\alpha^2\tf{q}{G}{}(\alpha)}\)
then it follows immediately that
\[
\ttf{I}{G}{1+\tf{r}{G}{}}=1+\ttf{J}{G}{} \,.
\]
\end{Theorem}
\begin{proof}
The last identity follows from the definition of \(\ttf{I}{G}{}\) and \(\ttf{J}{G}{}\) and from Corollary \ref{gamma2} :
\[
\ttf{I}{G}{1+\tf{r}{G}{}}-\ttf{J}{G}{} = \frac{\omega_\beta\,\beta^2\tf{p}{G}{}(\beta)}{\omega_\alpha\,\alpha^2\tf{q}{G}{}(\alpha)}-\frac{\omega_\gamma\gamma^2}{\omega_\alpha\,\alpha^2\tf{q}{G}{}(\alpha)}=\frac{\omega_\beta\,\beta^2\tf{p}{G}{}(\beta)}{\omega_\alpha\,\alpha^2\tf{q}{G}{}(\alpha)}+1-\frac{\omega_\beta\,\beta^2\tf{p}{G}{}(\beta)}{\omega_\alpha\,\alpha^2\tf{q}{G}{}(\alpha)}=1\,.
\]
That \(\ttf{A}{\alpha}{1}\) and \(\ttf{A}{\beta}{1}\) are well defined is clear from their definition
once it is shown that \(\ttf{I}{G}{}\) is indeed invariant (see Appendix \ref{sec:alphadiv}).
We see that the degree of \(\ttf{A}{\gamma}{2}\) equals
\( \omega_\beta -\omega_\alpha  -\deg(\tf{q}{G}{}) \)
for \(G=\bbbi, \bbbo, \bbbt\) and \(\bbbd_{n}\), \(n\) odd.
By inspection we see that this indeed equals zero.
For \(G=\bbbd_n\), one has \(\ttf{I}{\bbbd_n}{}=\frac{\beta}{\alpha^n}\) and we see that indeed
the degree of \(\alpha^{n-1}\) is equal to the degree of \(\gamma-2\), that is, \(2n-2\).

The invariance equations for \(G=\bbbt, \bbbo\) are
\begin{eqnarray*}
\bbbt&:& 2m_2  =m_1 \mod 3\,, \\
\bbbo&:& 1 = m_1 \mod 2\,,
\end{eqnarray*}
and one easily verifies from Table \ref{table_hom_ele} that they are satisfied.
For \(\bbbd_n\) they are
\[
m\, m_2 + m+2 = 2 m_1 \mod 4\,.
\]
One verifies that the chosen scaling indeed satisfies these equations.
\end{proof}
\begin{table}[ht!]
\begin{center}
\begin{tabular}{|c|c|c|c|} \hline
$G$  & \( \bbbi\) &  \( \bbbo\)  & \( \bbbt\) \\  
\hline \hline 
\(\ttf{I}{G}{}\) & \(\frac{1}{24490059264}\frac{\beta^3}{\alpha^5}\) & \(-\frac{1}{13500000}\frac{\beta^3}{\alpha^4}\)  & \(\frac{1}{2654208\,i\,\sqrt{3}}\frac{\beta^3}{\alpha^3}\) \\ [1ex]
\hline
\(\ttf{J}{G}{}\) & \(-\frac{1}{40479436800}\frac{\gamma^2}{\alpha^5}\) & \(\frac{1}{17280000}\frac{\gamma^2}{\alpha^4}\) & \(-\frac{1}{2359296\,i\,\sqrt{3}}\frac{\gamma^2}{\alpha^3}\) \\ [1ex]
\hline \hline
\( \ttf{A}{\alpha}{1}\) & \(\frac{1}{12\,\alpha}\tf{A}{\alpha}{1}\) & \(\frac{1}{6\,\alpha}\tf{A}{\alpha}{1}\) & \(\frac{1}{4\,\alpha}\tf{A}{\alpha}{1}\) \\ [1ex]
\hline 
\(\ttf{A}{\beta}{1}\) & \(\frac{1}{489801185280}\frac{\beta^2}{\alpha^5}\tf{A}{\beta}{1}\) &  \(-\frac{1}{108000000}\frac{\beta^2}{\alpha^4}\tf{A}{\beta}{1}\) & \(\frac{1}{10616832\,i\,\sqrt{3}}\frac{\beta^2}{\alpha^3} \tf{A}{\beta}{1}\) \\ [1ex]
\hline 
\(\ttf{A}{\gamma}{2}\) & \(-\frac{1}{586951833600}\frac{\beta}{\alpha^4}\tf{A}{\gamma}{2}\) & \(\frac{1}{190080000}\frac{\beta}{\alpha^3}\tf{A}{\gamma}{2}\) & \(-\frac{1}{17694720\,i\,\sqrt{3} }\frac{\beta}{\alpha^2}\tf{A}{\gamma}{2}\) \\ [1ex]
\hline
\end{tabular}
\end{center}
\caption{\(\alpha\) divisor: homogeneous elements of \(\bbbi, \bbbo, \bbbt\)}
\label{table_hom_ele}
\end{table}

\begin{center}
\begin{table}[ht!]
\begin{center}
\begin{tabular}{|c|c|c|} \hline
$G$& \(\bbbd_n\), \(n\) even  & \( \bbbd_{n}\), \(n\) odd\\ 
\hline \hline 
\(\ttf{I}{G}{}\) & \(\frac{\beta}{\alpha^{n}}\) & \(\frac{\beta^2}{\alpha^{n}}\) \\ [0.5ex]
\hline
\(\ttf{J}{G}{}\) &  \(\frac{1}{(2n)^2}\frac{\gamma^2}{\alpha^{2n}}\) & \(\frac{1}{n^2}\frac{\gamma^2}{\alpha^{n}}\) \\ [0.5ex]
\hline \hline
\( \ttf{A}{\alpha}{1}\) & \(\frac{1}{2\,\alpha}\tf{A}{\alpha}{1}\) & \(\frac{1}{2\,\alpha}\tf{A}{\alpha}{1}\)\\ [0.5ex]
\hline 
\(\ttf{A}{\beta}{1}\)  & \(\frac{1}{2n}\frac{1}{\alpha^{n}}\tf{A}{\beta}{1}\) & \(\frac{1}{n}\frac{\beta}{\alpha^{n}}\tf{A}{\beta}{1}\)\\ [0.5ex]
\hline 
\(\ttf{A}{\gamma}{2}\) &  \(\frac{1}{2(2n)^2(2n-1)}\frac{1}{\alpha^{n-1}}\tf{A}{\gamma}{2}\) & \(\frac{1}{2n^2(n-1)}\frac{\beta}{\alpha^{n-1}}\tf{A}{\gamma}{2}\)\\ [0.5ex]
\hline
\end{tabular}
\end{center}
\caption{\(\alpha\) divisor: homogeneous elements of \(\bbbd_n\), for \(n\) even and odd}
\label{DNtable_homalpha}
\end{table}
\end{center}

\begin{Corollary}
If one computes the trace-form among the basis elements one finds 
\begin{eqnarray*}
\begin{array}{ccc}\langle \ttf{A}{\alpha}{1},\ttf{A}{\alpha}{1}\rangle=2\,,&
\langle \ttf{A}{\alpha}{1},\ttf{A}{\beta}{1}\rangle=2\ttf{I}{G}{}\,,&
\langle \ttf{A}{\alpha}{1}, \ttf{A}{\gamma}{2}\rangle=-\ttf{I}{G}{}\,,
\\
\langle \ttf{A}{\beta}{1},\ttf{A}{\beta}{1}\rangle=2\ttf{I}{G}{2}\,,&
\langle \ttf{A}{\beta}{1}, \ttf{A}{\gamma}{2}\rangle=-\ttf{I}{G}{1-\tf{r}{G}{}}\,,&
\langle \ttf{A}{\gamma}{2}, \ttf{A}{\gamma}{2}\rangle=\frac{1}{2}\ttf{I}{G}{1-\tf{r}{G}{}}\,.
\end{array}
\end{eqnarray*}
\end{Corollary}

\begin{Remark}
It is clear from Tables \ref{table_hom_ele} and \ref{DNtable_homalpha}  that all homogeneous elements defined in (\ref{homogeneous1})--(\ref{homogeneous3}) have poles at the zeros of \(\alpha\) only. 
\end{Remark}
\begin{Remark}
The projective representation \(\sigma\) induces a linear representation on \(\bbbc(\lambda)\) and therefore a linear representation of \(G\) in \(\ttf{g}{}{}\). This implies that \(\frac{M}{f}\in\ttf{g}{}{G} \), that is, it is invariant under the action of \(G\).
\end{Remark}

\begin{Remark}[Towards Lax Pairs]
Defining a Lax operator \( L\) \(\in\ttf{g}{}{G}\,\) gives us a \( G\)--invariant (automorphic) Lax operator and therefore a \( G\)--invariant (automorphic) integrable systems of equations.
\end{Remark}

\begin{Example}[$\bbbd_n$, with $n=2$]\label{D_2ex5}
With reference to Example \ref{D_2ex3} one has
\begin{eqnarray*}
\rho(\ttf{A}{\alpha}{1})
&=& \left(
\begin{array}{cc}
0 &-\lambda\\
-\lambda^{-1} & 0
\end{array}
\right),\\
\rho(\ttf{A}{\beta}{1})
&=& \left(
\begin{array}{cc}
\frac{1-\lambda^{4}}{2\lambda^{2}} &   -\lambda^{-1}\\
-\lambda& -\frac{1-\lambda^{4}}{2\lambda^{2}}
\end{array}
\right),\\
\rho(\ttf{A}{\gamma}{2})
&=& 
\frac{1}{2} \left(\begin{array}{cc}
0 &  \lambda^{-1} \\
\lambda  & 0
\end{array}
\right).
\end{eqnarray*}
\eex
\end{Example}

The structure constants in the homogeneous case are given by the following theorem:
\begin{Theorem}\label{theo:commalpha}
The commutation relations for the basis of \(\ttf{\sl}{}{G}\) are
\begin{eqnarray}
{[\ttf{A}{\alpha}{1},\ttf{A}{\beta}{1}]}&=&2\,\ttf{I}{G}{}\ttf{A}{\alpha}{1}-2\,\ttf{A}{\beta}{1},     \label{eq:homcommabalpha}\\
{[\ttf{A}{\alpha}{1},\ttf{A}{\gamma}{2}]}&=&\ttf{A}{\beta}{1}+2\,\ttf{A}{\gamma}{2}\label{eq:homcommacalpha},\\
{[\ttf{A}{\beta}{1},\ttf{A}{\gamma}{2}]}&=&\ttf{I}{G}{1-\tf{r}{G}{}}\ttf{A}{\alpha}{1}+2\,\ttf{I}{G}{}\ttf{A}{\gamma}{2}\label{eq:homcommbcalpha}.
\end{eqnarray}
\end{Theorem}
\begin{proof}
This follows immediately from Theorems \ref{THEO:StuctureConstants} and \ref{theo:HomBasis}.
\end{proof}
\section{Normal form of the Lie algebra}\label{sec:normal}
In this section we derive the \emph{normal form} of the algebra  \(\ttf{g}{}{G}=\ttf{\sl}{}{G}\) (diagonalising it with respect to \(\ttf{A}{\alpha}{1}\)).

Let \[\mf{X}=x_1\ttf{A}{\alpha}{1}+x_2 \ttf{A}{\beta}{1}+ x_3 \ttf{A}{\gamma}{2}=\left(\begin{array}{c}x_1\\x_2\\x_3\end{array}\right).\]
Then
\begin{eqnarray*}
\ad(\ttf{A}{\alpha}{1})\mf{X}&=&x_2(2\,\ttf{I}{G}{}\,\ttf{A}{\alpha}{1}-2\,\ttf{A}{\beta}{1})
+x_3 (\ttf{A}{\beta}{1}+2\,\ttf{A}
    {\gamma}{2})
\\&=&\left(\begin{array}{c}2 x_2\,\ttf{I}{G}{}\\-2 x_2+x_3\\2 x_3\end{array}\right)
\\&=&
\left(
\begin{array}{ccc} 
0 & 2 \ttf{I}{G}{}& 0
\\
0&-2&1
\\
0&0&2
\end{array}\right)
\left(\begin{array}{c}x_1\\x_2\\x_3\end{array}\right)\,.
\end{eqnarray*}
This leads to the transformation matrix
\[
\left(\begin{array}{ccc} 1& 0 & 0 \\ -\ttf{I}{G}{}&1&0\\ \frac{1}{4}\ttf{I}{G}{}&\frac{1}{4}&1\end{array}\right)
\]
and suggests the definition of a new basis:
\begin{eqnarray}
\label{alphabasis1}
\ttf{e}{0}{}&=&\ttf{A}{\alpha}{1},\\ \label{alphabasis2}
\ttf{e}{-}{}&=&\ttf{A}{\beta}{1}-\ttf{I}{G}{}\, \ttf{A}{\alpha}{1},\\  \label{alphabasis3}
\ttf{e}{+}{}&=&\ttf{A}{\gamma}{2}+\frac{1}{4}\ttf{A}{\beta}{1}+\frac{1}{4}\ttf{I}{G}{}\,\ttf{A}{\alpha}{1}\,.
\end{eqnarray}
\begin{Corollary}
Computing the trace-form among the basis elements it results
\begin{eqnarray*}
\begin{array}{ccc}
\langle \ttf{e}{0}{},\ttf{e}{0}{}\rangle=2\,,&\langle \ttf{e}{0}{},\ttf{e}{-}{}\rangle=0\,,&
\langle \ttf{e}{0}{},\ttf{e}{+}{}\rangle=0\,,\\
\langle \ttf{e}{-}{},\ttf{e}{-}{}\rangle=0\,,&\langle \ttf{e}{-}{},\ttf{e}{+}{}\rangle=\ttf{I}{G}{1-\tf{r}{G}{}}\ttf{J}{G}{}\,,&\langle \ttf{e}{+}{},\ttf{e}{+}{}\rangle=0\,.
\end{array}
\end{eqnarray*}
\end{Corollary}
\begin{Example}[$\bbbd_n$, with $n=2$]\label{D_2ex6}
With reference to Example \ref{D_2ex5} we find
\begin{eqnarray} 
\rho(\ttf{e}{0}{})
&=&\left(\begin{array}{cc} 0& -\lambda\\
-\lambda^{-1} & 0\end{array}\right)\, ,\nonumber\\ 
\nonumber \\ 
\rho( \ttf{e}{-}{})
&=&\frac{1-\lambda^{4}}{2\lambda^{2}}
\left( \begin{array}{cc} 1 &-\lambda \\
\lambda^{-1} & -1 \end{array}\right) \,,\nonumber\\
\rho(\ttf{e}{+}{})
&=&\frac{1-\lambda^{4}}{8\lambda^{2}}\left(\begin{array}{cc} 1 & \lambda \\
-\lambda^{-1} & -1 \end{array}\right)\, \label{newx0y0h0}
 . \nonumber
\end{eqnarray}
See section \ref{sec:averaging} for comparison with the earlier results.
\eex
\end{Example}
It turns out that
\begin{eqnarray*}
{[\ttf{e}{+}{},\ttf{e}{-}{}]}&=&[\ttf{A}{\gamma}{2}+\frac{1}{4} \ttf{A}{\beta}{1}+\frac{1}{4}\ttf{I}{G}{} \,\ttf{A}{\alpha}{1}\,,\,\ttf{A}{\beta}{1}-\frac{1}{4}\ttf{I}{G}{}\, \ttf{A}{\alpha}{1}]
\\&=&
-[\ttf{A}{\beta}{1},\ttf{A}{\gamma}{2}]
+\ttf{I}{G}{}\, [\ttf{A}{\alpha}{1},\ttf{A}{\gamma}{2}]
+\frac{1}{2}\tf{I}{G}{}\,[\ttf{A}{\alpha}{1},\ttf{A}{\beta}{1}]
\\&=&
(- \ttf{I}{G}{1-\tf{r}{G}{}}+\ttf{I}{G}{2})\ttf{A}{\alpha}{1}
\\&=&
\ttf{I}{G}{1-\tf{r}{G}{}}(\ttf{I}{G}{1+\tf{r}{G}{}}-1)\ttf{A}{\alpha}{1}
\\&=&
\ttf{I}{G}{1-\tf{r}{G}{}}\,\ttf{J}{G}{}\,\ttf{e}{0}{}.
\end{eqnarray*}
Notice that \(Q(\ttf{I}{G}{})=\ttf{I}{G}{1-\tf{r}{G}{}}\,\ttf{J}{G}{}\) is a quadratic polynomial in \(\ttf{I}{G}{}\).
We check that
\begin{eqnarray*}
{[\ttf{e}{0}{},\ttf{e}{-}{}]}&=&
-2 \ttf{e}{-}{},
\\
{[\ttf{e}{0}{},\ttf{e}{+}{}]}&=& 2 \ttf{e}{+}{}.
\end{eqnarray*}
We have now proved the following theorem:
\begin{Theorem}\label{DTHEOalpha}
The \(G\)--Automorphic Lie algebras
\(\ttf{\sl}{}{G}\) are isomorphic as modules to \(\bbbc[\ttf{I}{G}{}]\otimes\sl\).
A basis for \(\ttf{\sl}{}{G}\) over \(\bbbc\),  is given by
\[
\ttf{e}{\cdot}{l}=\htf{I}{G}{l}\ttf{e}{\cdot}{}\,,\quad\cdot=0,\pm\,,
\quad l\in\mathbb{Z}_{\geq 0}\,,\]
where \(\htf{I}{G}{}=a\ttf{I}{G}{}+b\) for some \(a,b\in\bbbc\).
The commutation relations can be brought into the form
\begin{eqnarray*}
[\ttf{e}{+}{l_1},\ttf{e}{-}{l_2}]&=&\ttf{e}{0}{l_1+l_2}+\ttf{e}{0}{l_1+l_2+2}\\
{[\ttf{e}{0}{l_1},\ttf{e}{\pm}{l_2}]}&=&\pm2\ttf{e}{\pm}{l_1+l_2}
\end{eqnarray*}
The algebras are quasi--graded (e.g. \cite{lm_cmp05}), with grading depth \(2\).
\end{Theorem}
\begin{proof}
In all cases one has \([\ttf{e}{+}{0},\ttf{e}{-}{0}]=Q_G(\ttf{I}{G}{})\ttf{e}{0}{0}\)
where \(Q_G\) is a quadratic polynomial.
Using the allowed complex scaling and affine transformations on \(\ttf{I}{G}{}\) and rescaling \(\ttf{e}{-}{0}\),
this can be normalized to
\([\ttf{e}{+}{0},\ttf{e}{-}{0}]=(1+\htf{I}{G}{2})\ttf{e}{0}{0}\).
\end{proof}
\begin{Corollary}
The Automorphic Lie Algebras \(\ttf{\sl}{}{G}\) are isomorphic as Lie algebras for the groups \(\bbbt, \bbbo, \bbbi\) and \(\bbbd_n\).
\end{Corollary}
\begin{Remark}
This proves a conjecture by A. Mikhailov, made in 2008.
This conjecture is proven in \cite{PC09} by completely different methods,
unpublished at the time of writing of the present paper.
\end{Remark}
\begin{Definition}
A basis for \(\tf{\sl}{}{G}\),  is given by
\[
\ttf{e}{\cdot}{l}=\ttf{I}{G}{l}\ttf{e}{\cdot}{}\,,\quad\cdot=0,\pm\,,
\quad l\in\mathbb{Z}\,.\]
The commutation relations are
\begin{eqnarray*}
[\ttf{e}{+}{l_1},\ttf{e}{-}{l_2}]&=&\ttf{e}{0}{l_1+l_2+1}-\ttf{e}{0}{l_1+l_2+2}\\
{[\ttf{e}{0}{l_1},\ttf{e}{\pm}{l_2}]}&=&\pm2\ttf{e}{\pm}{l_1+l_2}
\end{eqnarray*}
\end{Definition}
\begin{Theorem}\label{DTHEOalpha0}
Let \(\tf{r}{G}{}=0\).
Then \(\tf{\sl}{}{G}=\tf{\sl}{+}{G}\oplus\tf{\sl}{-}{G}\) as a \(\bbbc[\ttf{I}{G}{},\ttf{I}{G}{-1}]\)-module, with subalgebras \(\tf{\sl}{\pm}{G}\) to be defined in the proof.
\end{Theorem}
\begin{proof}
Let \(\tf{\sl}{+}{G}=\langle \ttf{e}{0,\pm}{l}\rangle_{l\geq 0}\oplus \langle \ttf{e}{-}{-1}\rangle\).
The \(\bbbc[\ttf{I}{G}{}]\)-module  \(\tf{\sl}{+}{G}\) is a Lie algebra, since
\( [\ttf{e}{+}{0}, \ttf{e}{-}{-1}] =\ttf{e}{0}{0}-\ttf{e}{0}{1}
\in \tf{\sl}{+}{G}\).
Next consider
\(\tf{\sl}{-}{G}=\langle \ttf{e}{0,\pm}{l}\rangle_{l\leq -2}\oplus \langle \ttf{e}{+}{-1}, \ttf{e}{0}{-1}\rangle\).
This is also a Lie algebra, since \( [\ttf{e}{-}{-2}, \ttf{e}{+}{-1}]=
-\ttf{e}{0}{-2}+\ttf{e}{0}{-1}\in \tf{\sl}{-}{G}\) (here we only treated the worst commutators, the others are less critical).
Then \(\tf{\sl}{}{G}=\tf{\sl}{+}{G}\oplus\tf{\sl}{-}{G}\) as a \(\bbbc[\ttf{I}{G}{},\ttf{I}{G}{-1}]\)-module, with subalgebras \(\tf{\sl}{\pm}{G}\).
\end{proof}
\begin{Corollary}\label{cor:Mik}
The Automorphic Lie Algebras \(\tf{\sl}{}{G}\) are isomorphic as Lie algebras for the groups \(\bbbt, \bbbo, \bbbi\) and \(\bbbd_n\), \(n\) odd.
\end{Corollary}

\section{The example  \(\ttf{\sl}{}{\bbbd_2}\) and comparison with averaging}\label{sec:averaging}
We discuss here the case of the Automorphic Lie Algebra \(\ttf{\sl}{}{\bbbd_2}\) associated to \(\bbbd_2\); this was first found in \cite{lm_cmp05} via group average (and denoted as \(sl_{\bbbd_2}(2,\bbbc;0)\)) . This simple example allows us to show the equivalence of the two methods in this case.
Let 
\begin{eqnarray} 
&& x_0(\lambda ) 
=\left(\begin{array}{cc} 0& \lambda^{-1}\\
\lambda & 0\end{array}\right)\, ,\nonumber\\ 
\nonumber \\ 
&&y_0 (\lambda ) 
=\left(\begin{array}{cc} 0& \lambda \\
\lambda^{-1}& 0\end{array}\right)\, , \label{x0y0h0} \\ 
\nonumber \\ 
&&h_0 (\lambda )
=\frac{1-\lambda^4}{\lambda ^2}
\left( 
\begin{array}{cc} 1&0\\
0 & -1\end{array}\right) \,  \nonumber
\end{eqnarray}
be the generators of the algebra obeying the commutation relations 
\begin{eqnarray*}
{[x_{0} \,,\,y_{0}]}&=&h_{0}\\
{[h_{0}\,,\,x_{0}]}&=&(\tf{I}{G}{}+1)\,x_{0}-y_{0}\\
{[h_{0}\,,\,y_{0}]}&=&-(\tf{I}{G}{}+1)\,y_{0}+x_{0},
\end{eqnarray*}
where \(\tf{I}{G}{}=\frac{1}{2}\left(\frac{1+\lambda^{4}}{\lambda^{2}}\right)\) (the reader is referred to \cite{lm_cmp05} for details). To compare the results let us first of all write this algebra in normal form; this is equivalent to diagonalise the algebra with respect to \(y_0\) following the scheme used in Section \ref{sec:normal}. 
This leads to the transformation matrix
\[
\left(\begin{array}{ccc} 1& 0 & 0 \\ 2\,\tf{I}{G}{}&-2&1\\ -2\,\tf{I}{G}{}& 2 &1\end{array}\right)
\]
and suggests the definition of a new basis:
\begin{eqnarray}\label{diagbasis1}
\rho(\tf{e}{0}{})&=&-y_{0},\\ \label{diagbasis2}
\rho(\tf{e}{-}{})&=&\frac{1}{2}\left(-2\,\tf{I}{G}{}\,y_{0}+2\,x_{0}+h_{0}\right),\\ \label{diagbasis3}
\rho(\tf{e}{+}{})&=&\frac{1}{8}\left(2\,\tf{I}{G}{}\,y_{0}-2\,x_{0}+h_{0}\right)\,.
\end{eqnarray}
In this new basis the commutation relations read
\begin{eqnarray*}
{[\tf{e}{0}{},\tf{e}{\pm}{}]}&=&\pm 2 \tf{e}{\pm}{}\,,
\\
{[\tf{e}{+}{},\tf{e}{-}{}]}&=& (\tf{I}{G}{2}-1)\tf{e}{0}{}.
\end{eqnarray*}
The expressions \(\rho(\tf{e}{0}{})\),  \(\rho(\tf{e}{-}{})\) and \( \rho(\tf{e}{+}{})\) in (\ref{diagbasis1})--(\ref{diagbasis3}) are nothing but the generators \(\rho(\ttf{e}{0}{})\), \(\rho(\ttf{e}{-}{})\), \(\rho(\ttf{e}{+}{})\)  in Example \ref{D_2ex6}. 

\section{Explicit bases for the Automorphic Lie Algebras \(\ttf{\sl}{}{G}\)}\label{sec:explicit}
In this section we give explicit bases, using concrete formulas for the covariants.
We notice the remarkable likeness in all cases and remark upon the somewhat surprising fact that the determinant of all
the matrices is constant. We compute the Jordan normal form of the whole algebra and see that it is even more uniform in occurence. If understood, this might lead to better and quicker insight in the general case.
\begin{Theorem}\label{theo:big}
Let \(\tf{A}{0}{}, \tf{A}{\pm}{}\) be given by Remark \ref{rem:sl2}.
In \(\lambda\)--notation this reads as follows: \(\alpha_X=\alpha_\lambda\) and \(\alpha_Y=\omega_\alpha\alpha-\lambda\alpha_\lambda\), where \(\alpha_\lambda=\frac{d\alpha}{d\lambda}\).
Then
\begin{eqnarray}
\ttf{e}{0}{}&=&\frac{1}{\omega_\alpha\alpha}\tf{A}{0}{},\\
\ttf{e}{-}{}&=&-\frac{2}{\omega_\alpha\omega_\beta}\frac{\ttf{I}{G}{}\gamma}{\alpha\beta}\tf{A}{-}{},\\
\ttf{e}{+}{}&=&\frac{(\omega_\alpha-1)^2}{\omega_\alpha\omega_\beta}
\frac{\ttf{I}{G}{}\gamma}{\alpha\beta(\alpha,\alpha)^2}\tf{A}{+}{}
=-\frac{\omega_\beta}{2\omega_\alpha}
\frac{\beta}{\alpha\gamma}
\frac{\ttf{J}{G}{}}{\ttf{I}{G}{\tf{r}{G}{}}}
\tf{A}{+}{},\label{eq:34}
\end{eqnarray}
for all \(G\), where \(\tf{A}{-}{}=\tf{A}{}{}, \rho(\tf{A}{-}{})=\left(\begin{array}{cc}\lambda&-\lambda^2\\1&-\lambda\end{array}\right)\) and \(\tf{A}{0}{}=\tf{A}{\alpha}{1}\).
\end{Theorem}
\begin{proof}
Case by case inspection, using the results in the following sections.
The second equality in (\ref{eq:34}) uses Lemma \ref{lem:lem45}.
\end{proof}
In what follows we use the ${\lambda}$--notation
\begin{eqnarray*}
\rho(\tf{A}{+}{})&=&\left(\begin{array}{cc} \alpha_{\lambda}(\omega_{\alpha}\,\alpha-\lambda\,\alpha_{\lambda}) & (\omega_{\alpha}\,\alpha-\lambda\,\alpha_{\lambda})^2  \\
-\alpha^2_{\lambda} & -\alpha_{\lambda}(\omega_{\alpha}\,\alpha-\lambda\,\alpha_{\lambda})  \end{array}\right)\,,\\
\rho(\tf{A}{0}{})&=&\left(
\begin{array}{cc} \omega_{\alpha}\,\alpha-2\lambda\,\alpha_{\lambda} &  -2\lambda (\omega_{\alpha}\,\alpha-\lambda\,\alpha_{\lambda})\\
- 2\alpha_\lambda&  -\omega_{\alpha}\,\alpha+2\lambda\,\alpha_{\lambda}
\end{array}\right)\,.
\end{eqnarray*}

\subsection{Explicit basis for  \(\ttf{\sl}{}{\bbbi}\)}\label{sec:ExI}
Let  \(\omega_\alpha\) be \(12\);  it follows then that \(\omega_\beta=2\omega_\alpha-4=20\) and \(\omega_\gamma=3\omega_\alpha-6=30\); let also
\begin{eqnarray*}
\alpha(\lambda)&=&\lambda(\lambda^{10}+11\lambda^{5}-1)\,, \\ 
\alpha_{\lambda}(\lambda)&=&\frac{d\alpha(\lambda)}{d\lambda}=(11\lambda^{10}+66\lambda^{5}-1)\,,\\
\beta(\lambda)&=&-242 \left(\lambda^{20}-228 \lambda^{15}+494 \lambda^{10}+228 \lambda^5+1\right)\,, \\
\gamma(\lambda)&=&-4840 \left(\lambda^{30}+522 \lambda^{25}-10005 \lambda^{20}-10005 \lambda^{10}-522 \lambda^5+1\right)\,, \\
\ttf{I}{G}{}&=&\frac{1}{24490059264}\frac{\beta^3}{\alpha^5}\,.
\end{eqnarray*}
Then
\begin{eqnarray*}
\rho(\tf{A}{+}{})&=&\left(\begin{array}{cc} 
 11 \lambda-792 \lambda^6+4234 \lambda^{11}+792 \lambda^{16}+11 \lambda^{21} & 121 \lambda^2-1452 \lambda^7+4334 \lambda^{12}+132
\lambda^{17}+\lambda^{22} \\
 -1+132 \lambda^5-4334 \lambda^{10}-1452 \lambda^{15}-121 \lambda^{20} & -11 \lambda+792 \lambda^6-4234 \lambda^{11}-792 \lambda^{16}-11
\lambda^{21}
\end{array}\right)\,,\\
\rho(\tf{A}{0}{})&=&\left(
\begin{array}{cc} 
 -10 \lambda-10 \lambda^{11} & 22 \lambda^2-132 \lambda^7-2 \lambda^{12} \\
 2-132 \lambda^5-22 \lambda^{10} & 10 \lambda+10 \lambda^{11}
\end{array}\right)\,.
\end{eqnarray*}

\subsection{Explicit basis for  \(\ttf{\sl}{}{\bbbo}\)}\label{sec:ExO}
Let \(\omega_\alpha\) be \(6\);  it follows then that \(\omega_\beta=2\omega_\alpha-4=8\) and \(\omega_\gamma=3\omega_\alpha-6=12\); let also
\begin{eqnarray*}
\alpha(\lambda)&=&\lambda(\lambda^{4}-1)\,, \\ 
\alpha_{\lambda}(\lambda)&=&\frac{d\alpha(\lambda)}{d\lambda}=5\lambda^{4}-1\,,\\
\beta(\lambda)&=&-50 \left(\lambda^{8}+14 \lambda^4+1\right)\,, \\
\gamma(\lambda)&=&-400 \left(\lambda^{12}-33 \lambda^{8}-33 \lambda^4+1\right)\,,\\
\ttf{I}{G}{}&=&-\frac{1}{13500000}\frac{\beta^3}{\alpha^4}\,.
\end{eqnarray*}

Then
\begin{eqnarray*}
\rho(\tf{A}{+}{})&=&\left(
\begin{array}{cc}
 5 \lambda-26 \lambda^5+5 \lambda^9 & 25 \lambda^2-10 \lambda^6+\lambda^{10} \\
 -1+10 \lambda^4-25 \lambda^8 & -5 \lambda+26 \lambda^5-5 \lambda^9
\end{array}
\right)\,,\\
\rho(\tf{A}{0}{})&=&\left(
\begin{array}{cc}
 -4 \lambda-4 \lambda^5 & 10 \lambda^2-2 \lambda^6 \\
 2-10 \lambda^4 & 4 \lambda+4 \lambda^5
\end{array}
\right)\,.
\end{eqnarray*}

\subsection{Explicit basis for  \(\ttf{\sl}{}{\bbbt}\)}\label{sec:ExT}
Let \(\omega_\alpha\) be \(4\);  it follows then that \(\omega_\beta=2\omega_\alpha-4=4\) and \(\omega_\gamma=3\omega_\alpha-6=6\); let also
\begin{eqnarray*}
\alpha(\lambda)&=&\lambda^{4}-2\,i\,\sqrt{3}\,\lambda^{2}+1\,, \\ 
\alpha_{\lambda}(\lambda)&=&\frac{d\alpha(\lambda)}{d\lambda}=4\lambda(\lambda^2-\,i\,\sqrt{3})\,,\\
\beta(\lambda)&=&-96\,i\,\sqrt{3}\, \left(\lambda^{4}+2\,i\,\sqrt{3}\,\lambda^{2}+1\right)\,, \\
\gamma(\lambda)&=&-9216\lambda \left(\lambda^4-1\right)\,,\\
\ttf{I}{G}{}&=&\frac{1}{2654208\,i\,\sqrt{3}}\frac{\beta^3}{\alpha^3}\,.
\end{eqnarray*}
Then
\begin{eqnarray*}
\rho(\tf{A}{+}{})&=&\left(
\begin{array}{cc}
 -16 i \sqrt{3} \lambda-32 \lambda^3-16 i \sqrt{3} \lambda^5 & 16-32 i \sqrt{3} \lambda^2-48 \lambda^4 \\
 48 \lambda^2+32 i \sqrt{3} \lambda^4-16 \lambda^6 & 16 i \sqrt{3} \lambda+32 \lambda^3+16 i \sqrt{3} \lambda^5
\end{array}
\right)\,,\\
\rho(\tf{A}{0}{})&=&\left(
\begin{array}{cc}
 4-4 \lambda^4 & -8 \lambda+8 i \sqrt{3} \lambda^3 \\
 8 i \sqrt{3} \lambda-8 \lambda^3 & -4+4 \lambda^4
\end{array}
\right)\,.
\end{eqnarray*}

\subsection{Explicit basis for  \(\ttf{\sl}{}{\bbbd_n}\)}\label{sec:ExDn}
Let  \(\omega_\alpha\) be \(2\) and \(\omega_\beta\) be \(m\), then \(\omega_\gamma=\omega_\alpha+\omega_\beta-2=m\); recall that in this case \((\alpha,\alpha)^2=-2\neq\beta\). Let 
\begin{eqnarray*}
\alpha(\lambda)&=&\lambda\,,\\ 
\alpha_\lambda(\lambda)&=&1,\\
\beta(\lambda)&=&\frac{1}{2}(\lambda^{m}+1)\,, \\
\gamma(\lambda)&=&\frac{m}{2}(\lambda^{m}-1)\,, 
\end{eqnarray*}
where \(m=n\) if \(n\) is odd while \(m=2\,n\) if  \(n\) is even.  Moreover, \(\ttf{I}{G}{}=\frac{\beta}{\alpha^{n}}\) if  \(n\) is even while \(\ttf{I}{G}{}=\frac{\beta^2}{\alpha^{n}}\) if \(n\) is odd.
It turns out then
\begin{eqnarray*}
\rho(\tf{A}{+}{})&=&\left(
\begin{array}{cc}
\lambda & \lambda^2 \\
-1 & -\lambda
\end{array}
\right)\,,\\
\rho(\tf{A}{0}{})&=&-2\left(
\begin{array}{cc}
 0 & \lambda^2 \\
1 & 0
\end{array}
\right)\,.
\end{eqnarray*}

\begin{Remark}\label{rem:undressing}
One notices that \(\det( \rho(\ttf{e}{0}{}))=-1\) and \(\det( \rho(\ttf{e}{\pm}{}))=0\) in all cases.
We find that we can undress the representation, losing the invariance of the elements but keeping the commutation relations,
to the following form:
\begin{eqnarray*}
 \nu(\ttf{e}{0}{})&=&\left(\begin{array}{cc}1&0\\0&-1\end{array}\right),\\
 \nu(\ttf{e}{+}{})&=&
\frac{\omega_\beta}{2}
\frac{\ttf{J}{G}{}}{\ttf{I}{G}{\tf{r}{G}{}}}
\frac{\beta\alpha_\lambda}{\gamma}
\left(\begin{array}{cc}0&1\\0&0\end{array}\right),\\
 \nu(\ttf{e}{-}{})&=&\frac{2}{\omega_\beta}\ttf{I}{G}{}\frac{\gamma}{\beta\alpha_\lambda}\left(\begin{array}{cc}0&0\\1&0\end{array}\right).
\end{eqnarray*}
The undressing is done by computing the Jordan normal form of \( \rho(\ttf{e}{0}{})\) 
and applying the same conjugation to the other elements. That this is possible suggests that there must be a method to it.
The conjugation is with
\[
M=\left(\begin{array}{cc}-\lambda&\frac{\omega_\alpha\alpha-\lambda\alpha_\lambda}{\alpha_\lambda}\\
1&1\end{array}\right).
\]
\end{Remark}

\section{Conclusions}\label{sec:conclusion} 

We have shown that the problem of reduction can be formulated in a uniform way using the theory of invariants. This gives us  a powerful tool of analysis and it opens the road to new applications of these algebras, beyond the context of integrable systems. It turns out that in the explicit case we present here where the underlying Lie algebra is \(\sl\), we can compute  Automorphic Lie Algebras only using geometric data. Moreover we prove that Automorphic Lie Algebras associated to the groups \(\bbbt, \bbbo, \bbbi\) and \(\bbbd_n\) are isomorphic in the \(\Gamma_\alpha\) case.
This fact, i.e. that the Automorphic Lie Algebras are independent from the group is not
quite what one would expect from the topological point of view.
Indeed, if one divides out the group action one obtains usually an orbifold,
but a manifold in the case of \(\bbbi\). This distinction is not visible
at the level of the algebra and therefore not on the level of the
integrable systems that follow from the reduction procedure. It may turn
up again when one looks for the actual solutions, since then the domain
starts to play a role again. We leave this for further investigation.
On the other hand, the treatment of the groups, including \(\bbbz/n\), in the McKay-correspondence (see \cite{MR2500567})
and the resolution of the singularies of the relation between the invariants using invariant quotients of
the covariants \(\alpha,\beta\) and \(\gamma\) is remarkably uniform. 
We notice that the corresponding Dynkin diagram (without its weights) can be easily read off
from the degrees of \(\alpha,\beta\) and \(\gamma\) in Corollary \ref{gamma2} for each group, as long as \(G\neq\bbbz/n\).

Preliminary computations based on the icosahedral group \(\bbbi\) suggest that in the case of non equivalent \(\sigma\) and \(\tau_2\) one finds an Automorphic Lie Algebra isomorphic to the previous ones.

In the case of higher dimensional Lie algebras, say
\(\tf{sl}{k}{}(\bbbc)\), one could proceed as follows.
Let \(k\) be such that one of \(\bbbt, \bbbo, \bbbi\) has an irreducible projective representation \(\tau_k\).
Fix a \(2\)-dimensional irreducible representation \(\sigma\).
One can read off  the existence of an invariant matrix \(\tf{A}{}{i,0}\) of degree \(2i\) 
from the corresponding Dynkin diagram. Here \(\tf{A}{}{1,0}\) is the \(\mf{A}\) as used in this paper.
We plan to investigate these matters further.

%
%



\textbf{Acknowledgements}\\
The authors are grateful to A. V. Mikhailov for enlightening and fruitful discussions on various occasions.
One of the authors, S L, acknowledges financial support initially from EPSRC (EP/E044646/1) 
and then from NWO through the scheme VENI (016.073.026).

\appendix
\section{$\Zn{n}$, $\alpha$-divisor}\label{AppZN}
The analysis for the group \(\Zn{n}\) is different from the other groups
and so we give a completely independent treatment.
As before, let 
\[
\rho(\tf{A}{}{})=\left(\begin{array}{cc}XY&-X^2\\Y^2&-XY\end{array}\right).
\]
Take \(\alpha=X\) and \(\beta=Y\). Let \(m=n\) if \(n\) is odd and \(m=2n\) if \(n\) is even.
Let \(\omega \) be a \(m\)th root of unity.
The action of \(\Zn{m}\) is given by
\[
g X=\omega X,\quad g Y=\omega^{-1} Y.
\]
This induces an action of \(\Zn{n}\) on \(\lambda=Y/X\).
We see that \((\alpha,\beta)^1=-1\), so the previous setting does not apply.
In fact, \(\bbbc[X,Y]_{\Zn{m}}=\bbbc[\alpha,\beta]\), so the algebra of covariants is polynomial.
We can now compute the Hilbert function for \(\bbbc[\alpha,\beta]\otimes \tf{A}{}{}\)
and we see that it equals 
\[
\frac{3}{(1-t_a)(1-t_b)}
\]
where \(3\) stands for the \(3\)-dimensional space generated by
\((\alpha^2,\tf{A}{}{})^2, (\alpha\beta,\tf{A}{}{})^2\) and \((\beta^2,\tf{A}{}{})^2\),
in other words, for \(\sl\).
Let \(m=2d\). Then we take \(\ttf{I}{\Zn{2d}}{}=\frac{\beta^d}{\alpha^d}=\frac{\beta^n}{\alpha^n}\).
Let \(m=2d+1\).  Then we take \(\ttf{I}{\Zn{2d+1}}{}=\frac{\beta^m}{\alpha^m}=\frac{\beta^n}{\alpha^n}\).
We can draw the conclusion that \(\ttf{\sl}{}{\Zn{n}}=\bbbc[\ttf{I}{\Zn{n}}{}]\otimes\sl\).

\section{Invariant, $\alpha$-divisor}\label{sec:alphadiv}
In this Appendix we derive the $\alpha$-divisor invariants associated to each group. We consider expressions of the form
\(\frac{\beta^{m_2}}{\alpha^{m_1}}, m_1\in 
\bbbz_{>0}, m_2\in\bbbz_{\geq 0}\).
\subsection{$\bbbt$}\label{sec:appT}
To compute the invariants we have to solve the homogeneity equation 
\[
m_2 =m_1 \,,\quad m_1\in\bbbz_{>0},\quad m_2\in\bbbz_{\geq 0}\,, \label{eq1a}
\]
and the invariance equation
\[
2 m_2 =m_1 \,\mod 3\,. \label{eq2a}
\]
We take \(m_1=3\). This leads to \(\ttf{I}{\bbbt}{}\equiv \frac{\beta^3}{\alpha^3}\).
\subsection{$\bbbo$}\label{sec:appO}
To compute the invariants we have to solve the homogeneity and invariance equations 
\begin{eqnarray*}
4 m_2 &=&3 m_1 \,,\quad m_1\in\bbbz_{>0},\quad m_2\in\bbbz_{\geq 0}\,,\\
0&=&m_1 \,\mod 2 \,.
\end{eqnarray*}
We let \(m_1=2k_1\) and \(m_2=3k_2\). Then
\begin{eqnarray*}
2 k_2 &=& k_1 \,,\quad k_1\in\bbbz_{>0},\quad k_2\in\bbbz_{\geq 0}\,.
\end{eqnarray*}
Let \(k_1=2\), that is, \(m_1=4\) and \(k_2=1\), that is, \(m_2=3\).
Then \(\ttf{I}{\bbbo}{}\equiv\frac{\beta^3}{\alpha^4}\).
\subsection{$\bbbi$}\label{sec:appI}
To compute the invariants we have to solve the homogeneity  equation 
\begin{eqnarray*}
5 m_2 &=&3  m_1 \,,\quad m_1\in\bbbz_{>0},\quad m_2\in\bbbz_{\geq 0}\,.
\end{eqnarray*}
Let \(m_1=5k_1\) and \(m_2=3 k_2\). Then \(k_1=k_2\) and \(k_1>0\). We let \(k_1=k_2=1\), that is, 
\(m_1=5\) and \(m_2=3\). 
We have found \(\ttf{I}{\bbbi}{}\equiv\frac{\beta^3}{\alpha^5}\).
\subsection{$\bbbd_m$}\label{sec:appD}
To compute the invariants we have to solve the homogeneity equation 
\[
m\, m_2 = 2 m_1, \quad m_2\geq 0,\quad  m_1>0\,,
\]
and the invariance equation
\[
m\, m_2 = 2 m_1 \mod 4 \,.
\]
The last one follows from the first.
Let \(n=2d+p\), \(p=0,1\).
If \(n\) is even, we take \(m_1=n\) and \(m_2=1\).
If \(n\) is odd, we take \(m_1=n\) and \(m_2=2\).
Thus the invariant is
\(\ttf{I}{\bbbd_n}{}\equiv\frac{\beta^{1+p}}{\alpha^n}\).
\section{The $\beta$-divisor}\label{sec:betadivall}
We consider here the case of $\beta$-divisor; we use an underlined notation for this case. We consider expressions of the form
\(\frac{\alpha^{m_2}}{\beta^{m_1}}, m_1\in \bbbz_{\geq 0}, m_2\in\bbbz_{>0}\).

\subsection{Invariant, $\beta$-divisor}\label{sec:betadiv}
\subsubsection{$\bbbt$}\label{sec:appTbeta}
To compute the invariants we have to solve the homogeneity and invariance equations 
\begin{eqnarray*}
m_2 &=& m_1 \,,\quad m_1\in\bbbz_{>0},\quad m_2\in\bbbz_{\geq 0}\,, \label{eq1}\\
m_2 &=&2m_1 \,\mod 3\,.\label{eq2}
\end{eqnarray*}
We take \(m_2=3\). 
This leads to \(\ntf{I}{\bbbt}{}\equiv\frac{\alpha^3}{\beta^3}\).
\subsubsection{$\bbbo$}\label{sec:appObeta}
To compute the invariants we have to solve the homogeneity and invariance equations 
\begin{eqnarray*}
3 m_2 &=&4 m_1 \,,\quad m_1\in\bbbz_{>0},\quad m_2\in\bbbz_{\geq 0}\,,\\
0&=&m_2 \,\mod 2\,.
\end{eqnarray*}
We let \(m_2=2k_2\) and \(m_1=3k_1\). Then
\begin{eqnarray*}
k_2 &=& 2k_1 \,,\quad k_1\in\bbbz_{>0},\quad k_2\in\bbbz_{\geq 0}\,.
\end{eqnarray*}
Let \(k_1=1\), that is, \(m_1=3\) and \(k_2=2\), that is, \(m_2=4\).
Then \(\ntf{I}{\bbbo}{}\equiv\frac{\alpha^4}{\beta^3}\).

\subsubsection{$\bbbi$}\label{sec:appIbeta}
To compute the invariants we have to solve the homogeneity equation 
\begin{eqnarray*}
3 m_2 &=&5  m_1 \,,\quad m_1\in\bbbz_{>0},\quad m_2\in\bbbz_{\geq 0}\,.
\end{eqnarray*}
Let \(m_1=3k_1\) and \(m_2=5 k_2\). Then \(k_1=k_2\) and \(k_1>0\). We let \(k_1=k_2=1\), that is, 
\(m_1=3\) and \(m_2=5\). 
We have found \(\ntf{I}{\bbbi}{}\equiv\frac{\alpha^5}{\beta^3}\).
\subsection{$\bbbd_m$}\label{sec:appDbeta}
To compute the invariants we have to solve the homogeneity equation 
\[
2 m_2 = m\, m_1, \quad m_2\geq 0, m_1>0\,,
\]
and the invariance equation
\[
2 m_2 = m\, m_1 \mod 4\,.
\]
The last one follows from the first.
Let \(n=2d+p\), \(p=0,1\).
If \(n\) is even, we take \(m_2=n\) and \(m_1=1\).
If \(n\) is odd, we take \(m_2=n\) and \(m_1=2\).
Thus the invariant is
\(\ntf{I}{\bbbd_n}{}\equiv\frac{\alpha^n}{\beta^{1+p}}\).

\subsection{The homogeneous basis}\label{sec:homobasis_beta}
\begin{Lemma}\label{lem:HomBasisb}
The homogeneous basis of \(\ntf{\sl}{}{G}\), that is the \(G\)-Automorphic Lie Algebra based on \(\sl\) with poles in the zeros of \(\beta\),
is given by
\begin{eqnarray}\label{homogeneous1b}
\ntf{A}{\beta}{1}&=&
\ntf{I}{G}{}\ttf{A}{\beta}{1}\,, \\ \label{homogeneous2b}
\ntf{A}{\alpha}{1}&=&
\ntf{I}{G}{}\ttf{A}{\alpha}{1}\,,\\ \label{homogeneous3b}
\ntf{A}{\gamma}{2}&=&
\ntf{I}{G}{}\ttf{A}{\gamma}{2}\,,
\end{eqnarray}
where \(\ntf{I}{G}{}=\frac{\omega_\alpha}{\omega_\beta}\frac{\alpha^2\tf{q}{G}{}(\alpha)}{\beta^2\tf{p}{G}{}(\beta)} \) for
\(G=\bbbt, \bbbo, \bbbi\) and \(\bbbd_{n}\), \(n\) odd, and \(\ntf{I}{G}{2}=\frac{\omega_\alpha}{\omega_\beta}\frac{\alpha^2\tf{q}{G}{}(\alpha)}{\beta^2\tf{p}{G}{}(\beta)} \)
for \(G=\bbbd_{n}\), \(n\) even.
In other words, \(\frac{\omega_\alpha}{\omega_\beta}\frac{\alpha^2\tf{q}{G}{}(\alpha)} {\beta^2\tf{p}{G}{}(\beta)} =\ntf{I}{G}{1+\tf{r}{G}{}}\) where \(\tf{r}{G}{}=0\) for \(G=\bbbt, \bbbo, \bbbi,\bbbd_n\), \(n\) odd, and \(1\) for  \(\bbbd_{n}\), \(n\) even.
If we, moreover, define \(\ntf{H}{G}{}=\frac{\omega_\gamma\gamma^2}
{\omega_\beta\,\beta^2\tf{p}{G}{}(\beta)}\)
then it follows immediately that
\[
\ntf{I}{G}{1+\tf{r}{G}{}}=1-\ntf{H}{G}{} \,.
\]
\end{Lemma}
\begin{Corollary}
\( \ttf{I}{G}{}\ntf{I}{G}{}=1\) and \(\ttf{J}{G}{}=\ttf{I}{G}{} \ntf{H}{G}{}\).
\end{Corollary}
\begin{table}[h!]
\begin{center}
\begin{tabular}{|c|c|c|c|} \hline
$\beta$  & \( \bbbi\) &  \( \bbbo\)  & \( \bbbt\) \\  
\hline \hline 
\(\ntf{I}{G}{}\) & \(24490059264\,\frac{\alpha^5}{\beta^3}\) & \(-13500000\,\frac{\alpha^4}{\beta^3}\)  & \(2654208  \,i\,\sqrt{3}\frac{\alpha^3}{\beta^3}\) \\ [1ex]
\hline
\(\ntf{H}{G}{}\) & \(-\frac{121}{200}\frac{\gamma^2}{\beta^3}\) & \(-\frac{25}{32}\frac{\gamma^2}{\beta^3}\) & \(-\frac{9}{8}\frac{\gamma^2}{\beta^3}\) \\ [1ex]
\hline \hline
\( \ntf{A}{\alpha}{1}\) & \(2040838272\frac{\alpha^4}{\beta^3}\tf{A}{\alpha}{1}\) & \(-2250000\frac{\alpha^3}{\beta^3}\tf{A}{\alpha}{1}\) & \(663552 \,i\,\sqrt{3}\frac{\alpha^2}{\beta^3}\tf{A}{\alpha}{1}\) \\ [1ex]
\hline 
\(\ntf{A}{\beta}{1}\) & \(\frac{1}{20}\frac{1}{\beta}\tf{A}{\beta}{1}\) &  \(\frac{1}{8}\frac{1}{\beta}\tf{A}{\beta}{1}\) & \(\frac{1}{4}\frac{1}{\beta}\tf{A}{\beta}{1}\) \\ [1ex]
\hline 
\(\ntf{A}{\gamma}{2}\) & \(-\frac{121}{2900}\frac{\alpha}{\beta^2}\tf{A}{\gamma}{2}\) & \(-\frac{25}{352}\frac{\alpha}{\beta^2}\tf{A}{\gamma}{2}\) & \(-\frac{3}{20}\frac{\alpha}{\beta^2}\tf{A}{\gamma}{2}\) \\ [1ex]
\hline
\end{tabular}
\end{center}
\caption{\(\beta\) divisor: homogeneous elements of \(\bbbi, \bbbo, \bbbt\)}
\end{table}
\begin{center}
\begin{table}[h!]
\begin{center}
\begin{tabular}{|c|c|c|} \hline
$\beta$& \(\bbbd_n\), \(n\) even  & \( \bbbd_{n}\), \(n\) odd\\ 
\hline \hline 
\(\ntf{I}{G}{}\) & \(\frac{\alpha^{n}}{\beta}\) & \(\frac{\alpha^{n}}{\beta^2}\) \\ [0.5ex]
\hline
\(\ntf{H}{G}{}\) &  \(\frac{1}{(2n)^2}\frac{\gamma^2}{\beta^2}\) & \(\frac{1}{n^2}\frac{\gamma^2}{\beta^2}\) \\ [0.5ex]
\hline \hline
\( \ntf{A}{\alpha}{1}\) & \(\frac{1}{2}\frac{\alpha^{n-1}}{\beta}\tf{A}{\alpha}{1}\) & \(\frac{1}{2}\frac{\alpha^{n-1}}{\beta^2}\tf{A}{\alpha}{1}\)\\ [0.5ex]
\hline 
\(\ntf{A}{\beta}{1}\)  & \(\frac{1}{2n}\frac{1}{\beta}\tf{A}{\beta}{1}\) & \(\frac{1}{n}\frac{1}{\beta}\tf{A}{\beta}{1}\)\\ [0.5ex]
\hline 
\(\ntf{A}{\gamma}{2}\) &  \(\frac{1}{2(2n)^2 (2n-1)}\frac{\alpha}{\beta}\tf{A}{\gamma}{2}\) & \(\frac{1}{2n^2 (n-1)}\frac{\alpha}{\beta}\tf{A}{\gamma}{2}\)\\ [0.5ex]
\hline
\end{tabular}
\end{center}
\caption{\(\beta\) divisor: homogeneous elements of \(\bbbd_n\), for \(n\) even and odd}
\label{DNtable_hombeta}
\end{table}
\end{center}
\begin{Example}[$\bbbd_n$, with $n=2$]\label{D_2ex5beta}
In the case of $\bbbd_2$ one finds
\begin{eqnarray*}
\rho(\ntf{A}{\alpha}{1})
&=& \left(
\begin{array}{cc}
0 & -\frac{2\,\lambda^3}{1+\lambda^4}\\
-\frac{2\,\lambda}{1+\lambda^4} & 0
\end{array}
\right),\\
\rho(\ntf{A}{\beta}{1})
&=&  \left(
\begin{array}{cc}
\frac{1-\lambda^{4}}{1+\lambda^4} &   -\frac{2\lambda}{1+\lambda^4}\\
-\frac{2\lambda^3}{1+\lambda^4}& -\frac{1-\lambda^{4}}{1+\lambda^4}
\end{array}
\right),\\
\rho(\ntf{A}{\gamma}{2})
&=& \left(
\begin{array}{cc}
0 & \frac{\lambda}{1+\lambda^4}\\\
\frac{\lambda^3}{1+\lambda^4}\ & 0
\end{array}
\right)\,.
\end{eqnarray*}
\end{Example}
\begin{Theorem}\label{theo:commbeta}
The commutation relations for the basis of \(\ntf{\sl}{}{G}\) are
\begin{eqnarray}
{[\ntf{A}{\beta}{1},\ntf{A}{\alpha}{1}]}&=&2\,\ntf{I}{G}{}\ntf{A}{\beta}{1}-2\,\ntf{A}{\alpha}{1},     \label{eq:homcommabbeta}\\
{[\ntf{A}{\beta}{1},\ntf{A}{\gamma}{2}]}&=&\ntf{I}{G}{\tf{r}{G}{}}\ntf{A}{\alpha}{1}+2\,\ntf{A}{\gamma}{2}\label{eq:homcommacbeta},\\
{[\ntf{A}{\alpha}{1},\ntf{A}{\gamma}{2}]}&=&\ntf{I}{G}{}\ntf{A}{\beta}{1}+2\ntf{I}{G}{}\ntf{A}{\gamma}{2}\label{eq:homcommbcbeta}.  
\end{eqnarray}
\end{Theorem}
\begin{proof}
This follows immediately from Theorem \ref{theo:commalpha}.
\end{proof}
As for the case of \(\alpha\)-divisor one can now diaginalise the algebra to find its normal form; let \[\mf{X}=x_1\ntf{A}{\beta}{1}+x_2 \ntf{A}{\alpha}{1}+ x_3 \ntf{A}{\gamma}{2}=\left(\begin{array}{c}x_1\\x_2\\x_3\end{array}\right).\]
Then
\begin{eqnarray*}
\ad(\ntf{A}{\beta}{1})\mf{X}&=&x_2(2\,\ntf{I}{G}{}\ntf{A}{\beta}{1}-2\,\ntf{A}{\alpha}{1})
+x_3 (\ntf{I}{G}{\tf{r}{G}{}}\ntf{A}{\alpha}{1}+2\,\ntf{A}{\gamma}{2})
\\&=&\left(\begin{array}{c}2 x_2\,\ntf{I}{G}{}\\-2 x_2+x_3\,\ntf{I}{G}{\tf{r}{G}{}}\\2 x_3\end{array}\right)
\\&=&
\left(
\begin{array}{ccc} 
0 & 2 \ntf{I}{G}{}& 0
\\
0&-2&\ntf{I}{G}{\tf{r}{G}{}}
\\
0&0&2
\end{array}\right)
\left(\begin{array}{c}x_1\\x_2\\x_3\end{array}\right)\,.
\end{eqnarray*}
This leads to the transformation matrix
\[
\left(\begin{array}{ccc} 1& 0 & 0 \\ -\ntf{I}{G}{}&1&0\\ \frac{1}{4}\ntf{I}{G}{\tf{r}{G}{}+1}& \frac{1}{4}\ntf{I}{G}{\tf{r}{G}{}} &1\end{array}\right)
\]
and suggests the definition of a new basis:
\begin{eqnarray}
\label{betabasis1}
\ntf{e}{}{0}&=&\ntf{A}{\beta}{1},\\ \label{betabasis2}
\ntf{e}{}{-}&=&\ntf{A}{\alpha}{1}-\ntf{I}{G}{}\, \ntf{A}{\beta}{1},\\ \label{betabasis3}
\ntf{e}{}{+}&=&\ntf{A}{\gamma}{2}+\frac{1}{4}\ntf{I}{G}{\tf{r}{G}{}} \ntf{A}{\alpha}{1}+\frac{1}{4}\ntf{I}{G}{\tf{r}{G}{}+1}\,\ntf{A}{\beta}{1}\,.
\end{eqnarray}
In the new basis the commutation relations read
\begin{eqnarray*}
{[\ntf{e}{}{+},\ntf{e}{}{-}]}&=&-\ntf{I}{G}{}\,\ntf{H}{G}{}\,\ntf{e}{}{0}\,,\\
{[\ntf{e}{}{0},\ntf{e}{}{\pm}]}&=&\pm 2 \ntf{e}{}{\pm}.
\end{eqnarray*}
We have now proved the following theorem:
\begin{Theorem}\label{DTHEObeta}
The \(G\)--Automorphic Lie algebras
\(\ntf{\sl}{}{G}\) are isomorphic as modules to \(\bbbc[\ntf{I}{G}{}]\otimes\sl\).
A basis for \(\ntf{\sl}{}{G}\),  is given by
\[
\ntf{e}{l}{\cdot}=\ntf{I}{G}{l}\ntf{e}{}{\cdot}\,,\quad\cdot=0,\pm\,,
\quad l\in\mathbb{Z}_{\geq 0}\,,\]
The commutation relations can be brought into the form
\begin{eqnarray*}
[\ntf{e}{l_1}{+},\ntf{e}{l_2}{-}]&=&\ntf{e}{l_1+l_2+1}{0}-\ntf{e}{l_1+l_2+2+\tf{r}{G}{}}{0}\\
{[\ntf{e}{l_1}{0},\ntf{e}{l_2}{\pm}]}&=&\pm2\ntf{e}{l_1+l_2}{\pm}
\end{eqnarray*}
The algebras are quasi--graded (e.g. \cite{lm_cmp05}), with grading depth \(2+\tf{r}{G}{}\).
\end{Theorem}
\begin{Corollary}\label{cor:iso}
The Automorphic Lie Algebras \(\ntf{\sl}{}{G}\) are isomorphic as Lie algebras for the groups \(\bbbt, \bbbo, \bbbi\) and \(\bbbd_{n}\), with \(n\) odd, to the Automorphic Lie Algebras \(\ttf{\sl}{}{G}\).
\end{Corollary}
\begin{proof}
In all cases one has \([\ntf{e}{}{+},\ntf{e}{}{-}]=Q_G(\ntf{I}{G}{})\ntf{e}{}{0}\)
where \(Q_G\) is a quadratic polynomial.
Using the allowed complex scaling and affine transformations on \(\ntf{I}{G}{}\) and rescaling \(\ntf{e}{}{-}\),
this can be normalized to
\([\ntf{e}{}{+},\ntf{e}{}{-}]=(1+\htf{I}{G}{2})\ntf{e}{}{0}\).
\end{proof}
\begin{Remark}
In \cite[Section 3]{lm_cmp05} it was remarked that automorphic algebras corresponding to different orbits are not isomorphic, i.e. elements of one algebra cannot be represented by finite linear combination of the basis elements of the other algebra with complex constant coefficients. This seems to be in contradiction with Corollary \ref{cor:iso} for odd \(n\).
\end{Remark}

\end{document}